\documentclass[11pt]{article} 

\usepackage[margin=1in]{geometry}
\usepackage{float}
\usepackage{amssymb,amsfonts,amsmath}
\usepackage{enumerate}
\usepackage{enumitem}
\usepackage[hidelinks]{hyperref}
\usepackage{caption}
\usepackage{amsthm}
\usepackage{amsmath,amssymb,url}
\usepackage{color}
\usepackage{comment}

\usepackage[framemethod=TikZ]{mdframed}
\newcounter{Frame}
\newenvironment{Frame}[1][htb]{%
\refstepcounter{Frame}
    \begin{mdframed}[%
        frametitle={#1},
        skipabove=\baselineskip plus 2pt minus 1pt,
        skipbelow=\baselineskip plus 2pt minus 1pt,
        linewidth=1.0pt,
        frametitlerule=true,
    ]%
}{%
    \end{mdframed}
}

\newcommand{\sample}{S}

\newcommand{\stream}{\mathcal{S}}
\newcommand{\remove}[1]{}
\newcommand{\ttx}[1]{\texttt{#1}}

\newcommand{\eps}{\epsilon}
\newcommand{\tail}{\text{tail}}

\newcommand{\pr}[1]{\text{\normalfont Pr}\normalfont\left\lbrack #1 \right\rbrack} 
\newcommand{\ex}[1]{\mathbb{E}\normalfont\left\lbrack #1 \right\rbrack}

\newtheorem{claim}{Claim}

\newcommand{\p}{\mathfrak{p}} 

\usepackage[ruled,vlined]{algorithm2e}
\DontPrintSemicolon

\newcommand{\R}{\mathbb{R}}
\newcommand{\swor}{\texttt{SWOR}}

\newcommand{\saturated}{\texttt{saturated}}
\newcommand{\false}{\texttt{false}}
\newcommand{\true}{\texttt{true}}
\newcommand{\early}{\texttt{early}}
\newcommand{\poly}{\text{poly}}
\newcommand{\regular}{\texttt{regular}}
\newcommand{\thresh}{u}
\LinesNumbered
\theoremstyle{theorem}
\newtheorem{theorem}{Theorem}
\theoremstyle{lemma}
\newtheorem{lemma}{Lemma}
\theoremstyle{corollary}
\newtheorem{corollary}{Corollary}

\theoremstyle{proposition}
\newtheorem{proposition}{Proposition}
\theoremstyle{definition}
\newtheorem{definition}{Definition}

\title{Weighted Reservoir Sampling from Distributed Streams}

\author{
	Rajesh Jayaram\\
	Carnegie Mellon University\\
	\texttt{rkjayara@cs.cmu.edu}
	\and
	Gokarna Sharma \\
	Kent State University\\
	\texttt{gsharma2@kent.edu}
	\and 
	Srikanta Tirthapura\\
	Iowa State University\\
	\ttx{snt@iastate.edu}
	\and
	David P. Woodruff\\
	Carnegie Mellon University \\
	\texttt{dwoodruf@cs.cmu.edu}
}

\begin{document}
\maketitle

\begin{abstract}
We consider message-efficient continuous random sampling from a distributed stream, where the probability of inclusion of an item in the sample is proportional to a weight associated with the item. The unweighted version, where all weights are equal, is well studied, and admits tight upper and lower bounds on message complexity. For weighted sampling with replacement, there is a simple reduction to unweighted sampling with replacement. However, in many applications the stream has only a few heavy items which may dominate a random sample when chosen with replacement. Weighted sampling \textit{without replacement} (weighted SWOR) eludes this issue, since such heavy items can be sampled at most once.

In this work, we present the first message-optimal algorithm for weighted SWOR from a distributed stream. Our algorithm also has optimal space and time complexity.  
As an application of our algorithm for weighted SWOR, we derive the first distributed streaming algorithms for tracking \textit{heavy hitters with residual error}. Here the goal is to identify stream items that contribute significantly to the residual stream, once the heaviest items are removed. Residual heavy hitters generalize the notion of $\ell_1$ heavy hitters and are important in streams that have a skewed distribution of weights. In addition to the upper bound, we also provide a lower bound on the message complexity that is nearly tight up to a $\log(1/\eps)$ factor. Finally, we use our weighted sampling algorithm to improve the message complexity of distributed $L_1$ tracking, also known as count tracking, which is a widely studied problem in  distributed streaming. We also derive a tight message lower bound, which closes the message complexity of this fundamental problem. 

\end{abstract}
\sloppy
\maketitle

\section{Introduction}
We consider the fundamental problem of maintaining a random sample of a data stream that is partitioned into multiple physically distributed streams. In the streaming setting, it is often observed that many distributed monitoring tasks can be reduced to sampling. For instance, a search engine that uses multiple distributed servers can maintain the set of ``typical'' queries posed to it through continuously maintaining a random sample of all the queries seen thus far. Another application is network monitoring, one of the driving applications for data stream processing, which deploys multiple monitoring devices within a network. Each device receives extremely high rate data streams, and one of the most commonly desired aggregates is a random sample over all data received so far~\cite{DLT04,DLT03}.
Oftentimes, the predominant bottleneck in distributed data processing is the network bandwidth, and so it is highly desirable to have algorithms that communicate as few messages as possible. In particular, as the volume of the data and the number of distributed sites observing the streams scale, it becomes infeasible to communicate all of the data points observed to a central coordinator. Thus, it is necessary to develop algorithms which send significantly fewer messages than the total number of data points received.

The above applications have motivated the study of the \textit{continuous, distributed, streaming model~\cite{cmyz12}}, where there are $k$ physically distributed sites, numbered $1$ through $k$. Each site $i$ receives a local stream of data $\stream_i$. The sites are connected to a central coordinator, to which they can send and recieve messages. Queries are posed to the coordinator, asking for an aggregate over $\stream=\cup_{i=1}^{k} \stream_i$, the union of all streams observed so far. The goal is to minimize the {\em message complexity}, i.e., number of messages sent over the network, over the entire observation. \\

\noindent We now define the {\bf distributed weighted random sampling} problem. Each local stream $\stream_i$ consists of a sequence of items of the form $(e,w)$ where $e$ is an item identifier and $w$ is a positive weight. The streams $\stream_i$ are assumed to be disjoint. Let $n$ denote the size of $\stream = \cup_{i=1}^{k} \stream_i$. The task is for the coordinator to continuously maintain a weighted random sample of size $s=\min\{n,s\}$ from $\stream$. Note that the same identifier $e$ can occur multiple times, perhaps in different streams with different weights, and each such occurrence is to be sampled as if it were a different item. We consider two variations of weighted random sampling -- sampling with replacement and sampling without replacement. 

A single weighted random sample from $\stream$ is defined to be an item chosen from $\stream$ where {\em the probability of choosing item $(e, w)$ is proportional to $w$, i.e., equal to $\frac{w}{\sum_{(e',w') \in \stream}w'}$.} 
\begin{definition}
\label{def:wswor}
A \textit{weighted random sample without replacement } \allowbreak (weighted SWOR) from $\stream$ is a set $S$ generated according to the following process. Initially $S$ is empty. For $i$ from $1$ to $s$, a single weighted random sample is chosen from $(\stream \setminus S)$ and added to $S$.
\end{definition}

\begin{definition}
\label{def:wswr}
A \textit{weighted random sample with replacement} (weighted SWR) from $\stream$ is a set $S$ generated according to the following process.
Initially $S$ is empty. For $i$ from $1$ to $s$, a single weighted random sample is chosen from $\stream$ and added to $S$. 
\end{definition}

\begin{definition} \label{def:WSWOR}
A distributed streaming algorithm $\mathcal{P}$ is a weighted sampler without (with) replacement if for each $t > 0$, the coordinator maintains a set $S$ of size $\min\{t,s\}$ such that $S$ is a weighted random sample chosen without (with) replacement from all items seen so far, $\{(e_1,w_1),\dots,(e_t,w_t)\}$.
\end{definition}

Distributed random sampling generalizes the classic {\em reservoir sampling} problem \cite{Knuth,Vitter85} from a centralized to a distributed setting. Random sampling serves as a building block in other aggregation tasks, such as estimation of the number of distinct elements \cite{Cormode:2005,Cormode:2011} and identifying heavy hitters \cite{Babcock:2003,Keralapura:2006,Manjhi:2005,Yi:2013}. Distributed random sampling has also been used in approximate query processing in big data systems such as BlinkDB \cite{BlinkDB}. Unweighted distributed sampling, where all weights are equal, is well studied and admits tight upper and lower bounds on message complexity \cite{TW11,cmyz12,CTW16}. However, to the best of the authors' knowledge, weighted distributed sampling has not been studied so far, though there is prior work in the centralized setting \cite{ES06,BOV15}.

Designing a message-efficient distributed sampler is non-trivial. One challenge is that the system state is distributed across multiple sites. It is not possible to keep the distributed state tightly synchronized, since this requires a significant message overhead. Since the states of the sites are not synchronized with the coordinator, sites may send updates to the coordinator even though the sample at the coordinator does not change. A second challenge is that the traditional metrics for a centralized sampling algorithm, such as space complexity and time per item, do not affect the message complexity of a distributed algorithm. For instance, a centralized algorithm that uses $O(1)$ update time and optimal $O(s)$ space need not lead to a distributed algorithm with optimal messages, since the number of messages sent depends on how many times the random sample changes according to the views of the sites and the coordinator. Finally, we emphasize the fact that, in Definition \ref{def:WSWOR}, the protocol must maintain a weighted SWOR \textit{at all times} in the stream.  There is no notion of failure in the definition of a weighted sampler -- the protocol can never fail to  maintain the sample $S$. These two features make the problem substantially more challenging.
\subsection{Contributions}
Let $k$ denote the number of sites, $s$ the desired sample size, and $W$ the total weight received across all sites.

\medskip
\noindent $\bullet$ We present an algorithm for weighted SWOR that achieves an optimal expected message complexity of $O\left( k \frac{\log(W/s)}{\log(1+k/s	)} \right)$.
The algorithm uses an optimal $\Theta(1)$ space at each site, and an optimal $\Theta(s)$ space at the coordinator. The update time of our algorithm is also optimal: $\Theta(1)$ for a site to process an update, and $O\left( k \frac{\log(W/s)}{\log(1+k/s	)} \right)$
total runtime at the coordinator. Our algorithm is the first message-optimal algorithm for distributed weighted SWOR. We note that a message-efficient algorithm for weighted SWR follows via a reduction to the unweighted case, and obtaining an algorithm for SWOR is significantly harder, as described below.

\medskip

\noindent $\bullet$ As an application of our weighted SWOR, we provide the first distributed streaming algorithm with small message complexity for continuously monitoring heavy hitters with a \textit{residual} error guarantee. This allows us to identify heavy hitters in the residual stream after extremely heavy elements are removed. A residual error guarantee is stronger than the guarantee provided by $\ell_1$ heavy hitters, and is especially useful for streams where the weight of items is highly skewed~\cite{BICS10}. The expected message complexity of our algorithm is $O(k \log(W)/ \log(k) + \log(1/\eps)\log(W)/\eps)$. We prove that our algorithm is nearly optimal, by also giving a $\Omega(k \log(W)/\log(k) + \log(W)/\eps)$ lower bound, which is tight up to a $\log(1/\eps)$ factor in the second term. 

\medskip

\noindent $\bullet$ We demonstrate another application of our sampling algorithms to the well-studied problem of $L_1$ tracking (or count tracking) in distributed streams, which requires the coordinator to maintain a $(1\pm \eps)$ relative error approximation of the total weight seen so that at any given time, with constant probability, the estimate is correct. 

For the case $k \geq 1/\eps^2$, the best known upper bound was $O(k \log W)$ messages in expectation \cite{huang2012randomized}. Our algorithm for $L_1$ tracking uses $O(  k\log(W)\allowbreak/\log(k)\allowbreak  + \log(\eps W)/\eps^2)$ messages in expectation, which improves on the best known upper bound when $k \geq 1/\eps^2$. In this setting, we also improve the lower bound from $\Omega(k)$ to $\Omega(k \frac{\log(W)}{\log(k)})$. 

For the case $k \leq 1/\eps^2$, matching upper and lower bounds of $\Theta((\sqrt{k}/\eps)\log W)$ were known \cite{huang2012randomized}. When $k \geq 1/\eps^2$, the lower bound of \cite{huang2012randomized} becomes $\Omega(\log(W)/\eps^2)$. So if $k \geq 1/\eps^2$ and $ k\frac{\log(W)}{\log(k)} < \log(W)/\eps^2$, our upper bound is $O(\log(W)/\eps^2)$, which is tight, and otherwise our upper bound is $O(k \log(W)/\log(k))$, which is also tight. Thus, combined with the upper and lower bounds of \cite{huang2012randomized}, our results close the complexity of the distributed $L_1$ tracking problem.

\subsection{Our Techniques}
For the problem of sampling \textit{with} replacement, there is a relatively straightforward reduction from the weighted to the unweighted case, which we elaborate on in Section \ref{sec:basic}. The reduction involves duplicating a item $(e,w)$ a total of $w$ times into \textit{unweighted} updates, and does not require an increase in message complexity. On the other hand, there are inherent difficulties in attempting to carry out a similar reduction for sampling \textit{without} replacement, which we will now discuss.  

\vspace{2mm} 
\noindent{\bf On the Difficulty of a Reduction from Weighted SWOR to Unweighted SWOR:~} 
We now examine the difficulties which arise when attempting to reduce the problem of weighted SWOR to unweighted SWOR. We consider the following natural candidate reduction.
Given a weighted stream of items $\stream = \{(e_i,w_i)| i= 1 \ldots n\}$, consider an unweighted stream $\stream'$ where for each $(e_i,w_i) \in \stream$, there are $w_i$ copies of $e_i$ in $\stream'$. Note that $\stream'$ can be constructed in a streaming manner as items of $\stream$ are seen. 

Let $S'$ be an unweighted SWOR of $s$ items from $\stream'$. We remark that \textit{if $S'$ consists of $s$ distinct identifiers}, then those distinct identifiers are in fact a weighted SWOR of $\stream$. The difficulty of using this reduction is that of ensuring $s$ distinct items within $S'$. One could consider a method that maintains an unweighted SWOR of size greater than $s$ in $S'$, expecting to get at least $s$ distinct identifiers among them. However, this is not straightforward either, due to the presence of heavy-hitters (items with very large weight) which may contribute to a large fraction of the total weight in $\stream$. For instance, if there are $s/2$ items that contribute to a more than $1-1/(100s)$ fraction of the total weight within $\stream$, then $S'$ is likely to contain only identifiers corresponding to these items. This makes it very unlikely that $S'$ has $s$ distinct identifiers, even if the size of $S'$ is much larger than $s$. If the number of distinct items in $S'$ falls below $s$, then one could invoke a ``recovery'' procedure that samples further items from the stream to bring the sample size back to $s$, but this itself will be a non-trivial distributed algorithm. Moreover, re-initializing or recovering the algorithm would be costly in terms of message complexity, and introduce unintended conditional dependencies into the distribution of the output. Note that one cannot apply distinct sampling~\cite{GT01,Gibbons01} to maintain $s$ distinct items from $\stream'$, since distinct sampling completely ignores the frequency of identifiers in $\stream'$ (which correspond to weights in $\stream$), while we do not want this behavior -- items with a greater weight should be chosen with a higher probability. 

Thus, one of the primary difficulties with an algorithm for weighted SWOR is to handle heavy hitters. One could next consider a method that explicitly maintains heavy hitters within the stream (using a streaming algorithm). On the remainder of the stream excluding the heavy hitters, one could attempt to apply the reduction to unweighted SWOR as described above. However, this does not quite work either, since after removing the heavy hitters from the stream, the remaining weight of the stream may still be dominated by just a few items, and the same difficulty persists. One has only swapped one set of heavy hitters for another. Such ``residual heavy hitters'' may not be heavy hitters in the original stream, and may have evaded capture when the first heavy hitters were identified. This may proceed recursively, where the weight of the stream is still dominated by only a few items after removing the heavy hitters and the residual heavy hitters. Therefore, heavy hitter identification does not solve our problem and further ideas are needed.

\vspace{2mm}

\noindent{\bf Algorithmic Ideas for Weighted SWOR:} Our main algorithm for weighted SWOR combines two key ideas. 
Our first key technical contribution is to divide the items into multiple ``level sets'' according to the magnitude of their weights. All items within a single ``level set'' have weights that are close to each other. Our algorithm withholds an item from being considered by the sampler until the level set that the item belongs to has a (pre-specified) minimum number of items. This property ensures that when an item is considered for sampling, there are at least a minimum number of items of the same (or similar) weight, so that the difficulty that we faced with extreme heavy hitters does not surface. When a level set reaches a certain size, all items within the level set are released for further sampling. By choosing the weights of different levels to increase geometrically, we ensure that the number of level sets remains small, and the overhead of additional messages for filling up the level sets is small. While an item is withheld, we also run a procedure so that, at every time step, it will still be output in the sample at that time with the correct probability. Thus, the notion of withholding an item applies means only that it is withheld from an internal sampling algorithm (a subroutine), and not the overall sampler. Note that, for a distributed sampler to be correct, it cannot entirely withhold an item from being sampled, even for a single time step. 


The second idea is that of ``precision sampling''. Originally used in \cite{andoni2010streaming} for sampling in non-distributed data streams, and then extended by \cite{Jowhari:2011, jayaram2018perfect}, the idea of precision sampling is to scale each weight $w_i$ by an i.i.d. random variable $x_i$, which generates a ``key'' $v_i = w_i x_i$ for each item $(e_i,w_i)$. One then looks at the resulting vector $v = (v_1,v_2,\dots,v_n)$ of keys, and returns the largest coordinates in $v$. In particular, \cite{jayaram2018perfect} used the random variables $x_i = 1/t_i^{1/p}$ where $t_i$ is \textit{exponentially distributed}, to develop perfect $L_p$ samplers. A similar idea is also used in ``priority sampling''~\cite{DLT07} which was developed in the context of network monitoring to estimate subset sums. We use precision sampling to generate these keys for each item, such that the items with the $s$ largest keys form the weighted SWOR of the stream. 

It is not difficult to see that if each site independently ran such a sampler on its input--storing the items with the $s$ largest keys--and sent each new sample to the coordinator, who then stores the items with the overall $s$ largest keys, one would have a correct protocol with $O(ks\log(W))$ expected communication. This could be carried out by generating the keys $w_i/t_i$ with exponential variables $t_i$, or also by using the keys $u_i^{1/w_i}$ with $u_i$ uniform on $(0,1)$, as used for non-distributed weighted sampling without replacement in \cite{ES06}. Thus, a key challenge addressed by our work is to improve the na\"ive multiplicative bound of $\tilde{O}(ks)$ to an additive bound of  $\tilde{O}(k+s)$.


We also remark that by duplicating weighted items \textit{in combination} with our new level set technique, it may be possible to adapt the message-optimal \textit{unweighted} sampling algorithms of \cite{TW11,cmyz12} (see also~\cite{TW19-arxiv}) to yield a weighted SWOR. The approach would nonetheless require the use of level sets, along with a similar analysis as provided in this work, to remove extremely heavy items from the stream which would cause issues with these samplers. We believe that our approach by scaling an entire weight by a random variable, rather than manually duplicating weighted items $(e,w)$ into $w$ unweighted items, is more natural and simpler to understand. Moreover, it is likely that we obtain improved runtime bounds at the sites by using the algorithm presented in this paper (which are runtime optimal). 

\vspace{2mm}
 
\noindent{\bf Residual Heavy Hitters:}
For the problem of monitoring heavy hitters, the technique of sampling \textit{with replacement} has been frequently applied. By standard coupon collector arguments, taking $O(\log(1/\eps)/\eps)$ samples with replacement is enough to find all items which have weight within an $\eps$ fraction of the total. On the other hand, it is possible and frequently the case that there are very few items which contain nearly all the weight of the stream. This is precisely the domain where SWOR achieves remarkably better results, since a with replacement sampler would only ever see the heavy items. Using this same number of samples \textit{without replacement}, we demonstrate that we can recover all items which have weight within an $\eps$ fraction of the total \textit{after} the top $1/\eps$ largest items are removed. This is \textit{residual error} guarantee is much stronger in the case of skewed distributions of weights, and is a important application of SWOR.

\vspace{2mm}
 
\noindent{\bf $L_1$ Tracking:} Finally, we observe that the following desirable property of our weighted SWOR: namely that, once the heavy hitters are withheld, the values of the keys stored by the algorithm at any time provide good estimates of the total $L_1$ (the sum of all the weights seen so far). By taking enough samples, we can show that the $s$-th order statistic, that is the $s$-largest key overall, concentrates around the value $W^t/s$, up to a $(1 \pm \eps)$ factor, where $W^t$ is the sum of all the weights up to a given time $t$. Unfortunately, withholding heavy items alone is not sufficient for good message complexity, as one must naively withhold an extra $\Theta(1/\eps)$ factor more heavy items than the size of $s$ to obtain good concentration. To avoid a $\Theta(1/\eps)$ blow-up in the complexity, we must remove heavy hitters in another way. To do this, we duplicate updates instead of withholding them, a trick used in \cite{jayaram2018perfect} for a similar sampler. This observation yields an optimal algorithm for distributed $L_1$ tracking for $k \geq 1/\epsilon^2$, by using our weighted SWOR as a subroutine. Previously an optimal algorithm for $L_1$ tracking was known only for $k < 1/\epsilon^2$. 

\subsection{Related Work}
Random sampling from a stream is a fundamental problem and there has been substantial prior work on it. The reservoir sampling algorithm (attributed to Waterman \cite{Knuth}) has been known since the 1960s. There has been much follow-up work on reservoir sampling including methods for speeding up reservoir sampling~\cite{Vitter85}, sampling over a sliding window~\cite{BOZ12,GT02,XTB08,BDM02,Gemulla:2008}, and sampling from distinct elements in data~\cite{GT01,Gibbons01}.

The sequential version of weighted reservoir sampling was considered by Efraimidis and Spirakis~\cite{ES06}, who presented a one-pass $O(s)$ algorithm for weighted SWOR. Braverman et al.~\cite{BOV15} presented another sequential algorithm for weighted SWOR, using a reduction to sampling with replacement through a ``cascade sampling'' algorithm.
Unweighted random sampling from distributed streams has been considered in prior works~\cite{cmyz12,CTW16,TW11}, which have yielded matching upper and lower bounds on sampling without replacement. Continuous random sampling for distinct elements from a data stream in a distributed setting has been considered in \cite{CT15}.

There has been a significant body of research on algorithms and lower bounds in the continuous distributed streaming model. This includes algorithms for frequency moments, \cite{Cormode:2011,Cormode:2005}, entropy estimation \cite{Arackaparambil:2009,ChenZ17}, heavy hitters and quantiles ~\cite{Yi:2013}, distributed counts~\cite{huang2012randomized}, and lower bounds on various statistical and graph aggregates~\cite{WoodruffZ17}.  



\vspace{2mm}

\noindent {\bf Roadmap:} We present preliminaries and basic results in Section~\ref{sec:prelims}, followed by an optimal algorithm for weighted SWOR in Section~\ref{sec:SWORexpo}, applications to residual heavy hitters and lower bounds in  Section~\ref{sec:reshh}, applications to $L_1$ tracking and lower bounds in Section~\ref{sec:l1}, and concluding remarks in Section~\ref{sec:conc}.


\section{Preliminaries and Basic Results}
\label{sec:prelims}

\subsection{Preliminaries}
As in prior work in the continuous distributed streaming model, we assume a {\em synchronous} communication model, where the system operates in {\em rounds}, and in each round, each site can observe (at most) one item, and send a message to the coordinator, and receive a response from the coordinator. We assume that the messages are delivered in FIFO order (no overtaking of messages), there is no message loss, and the sites and the coordinator do not crash. We make no assumption on the sizes of the different local streams received by different sites, the order of arrival, and the interleaving of the streams at different sites. The only assumption we have is a global ordering of the stream items by their time of arrival onto one of the sites. If $n$ is the total number of items observed in the system, then for $j = 1 \ldots n$, $o^j = (e_j,w_j)$ is the $j$th item observed in the total order. The partitioning of the items across different processors is carried out by an adversary.

Our space bounds are in terms of machine words, which we assume are of size $\Theta(\log(nW))$ bits, and that arithmetic operations on machine words can be performed in $O(1)$ time. We also assume that an element identifier and weight fits in a constant number of words. In our algorithms, the length of a message is a constant number of words, so that the number of messages is of the same order as the number of words communicated over the network. For simplicity (but without loss of generality), in the following sections we assume each weight $w_j$ satisfies $w_j \geq 1$. Since each weight $w_j$ can be written in a constant number of machine words by assumption, it follows that we could always scale the weights by a polynomial factor to ensure $w_j \geq 1$, which blows up the complexity of our algorithms only by a constant, since we only have logarithmic dependency on the total weight. 




\subsection{Basic Results}\label{sec:basic}
We first present some basic results for weighted SWR and weighted SWOR. Let $n$ denote the number of items received in the stream. 

\medskip

\noindent{\bf Algorithm for Weighted SWR:} We first derive a message-efficient algorithm for weighted SWR using a reduction to unweighted SWR. For this reduction, we assume that the weights $w_i$ are integers. We first note the following result, from~\cite{cmyz12}.

\begin{theorem}[\cite{cmyz12}]
\label{thm:unweighted-swr-cmyz}
There is a distributed algorithm for unweighted SWR with message complexity $O((k + s \log s) \frac{\log n}{\log (2+k/s)} )$. The coordinator needs $O(s)$ space and $O((k + s \log s) \frac{\log n}{\log(2+k/s)} )$ total time; each site needs $O(1)$ space and $O(1+ \frac{s}{n}\log s \frac{\log n}{\log(2+k/s)} )$ time per item amortized.
\end{theorem}

\begin{corollary}
\label{cor:wswor}
There is a distributed algorithm for weighted SWR with message complexity $O((k + s \log s) \frac{\log W}{\log (2+k/s)} )$ where $W$ is the total weight received so far. The coordinator needs $O(s)$ space and $O((k + s \log s) \frac{\log W}{\log(2+k/s)} )$ total time. Each site needs $O(1)$ space and the amortized processing time per item at a site is 
$O(1+ \frac{1}{n}(k + s \log s) \frac{\log W}{\log (2+k/s)} )$.
\end{corollary}

\begin{proof}
We reduce weighted SWR to unweighted SWR as follows. Given stream of items with weights $\stream = (e_i,w_i), i= 1 \ldots n$, consider an unweighted stream $\stream'$ where for each $(e_i,w_i) \in \stream$, there are $w_i$ copies of $e_i$ in $\stream'$. Note that $\stream'$ can be constructed in a streaming manner as items of $\stream$ are seen. 

Note that an unweighted SWR of $s$ items chosen from $\stream'$ is a weighted SWR of $s$ items chosen from $\stream$. To prove this, let $S'$ denote an unweighted SWR of size $s$ from $\stream'$. Consider an arbitrary item in $S'$. Since there are $w_i$ copies of $e_i$ in $\stream'$, this item is likely to be $e_i$ with probability $w_i/\sum_{j}w_j$, and hence obeys the distribution we desire of a single item in a weighted SWR of $\stream$. Since the items of $S'$ are chosen independent of each other, the entire sample $S'$ has the same distribution as a weighted SWR of $s$ items from $\stream$. The number of items in $\stream'$ is equal to $W$, the total weight of $\stream$. The message complexity, as well as the space complexity at the sites and the coordinator follow from Theorem~\ref{thm:unweighted-swr-cmyz}. 

The processing time per item at the site needs additional work. An unweighted SWR sampler simulates the execution of $s$ independent copies of a single element sampler. Each element $(e,w)$ in the weighted input stream leads to $w$ elements into each of the $s$ samplers. Naively done, this reduction leads to a total runtime of $O(sw)$, which can be very large. To improve on this, we speed up the sampling as follows. We note that in the algorithm in Section 3.2 of \cite{cmyz12} for a single element unweighted SWR, in round $j$, an element is sent to the coordinator with probability $2^{-j}$ \footnote{The algorithm in \cite{cmyz12} divides execution into rounds; the exact definition of a round does not impact our analysis here, and is hence not provided.} When $w$ elements are inserted, as in our reduction, the probability that at least one of them is sent to the coordinator is $\alpha(w,j) = 1-{(1-2^{-j})}^w$. The site can simply send the element $(e,w)$ to the coordinator with probability $\alpha(w,j)$. Repeating this $s$ times leads to a runtime of $O(s)$ per element. To further improve on this, note that across all the $s$ single element samplers, the number of samplers that send a message to the coordinator is a binomial random variable $B(s,\alpha(w,j))$. In the improved algorithm, the site samples a number $X$ from this distribution. If $X > 0$, it chooses a random size $X$ subset of the $s$ samplers, and send $(e,w)$ to the coordinator for each chosen sampler -- as also noted in \cite{cmyz12}, this leads to the same distribution as making an independent decision for each sampler. The total time taken at the sites is now of the same order as the number of messages sent to the coordinator, which is $O\left((k + s \log s) \frac{\log W}{\log (2+k/s)} \right)$. The amortized time per element at the site is thus $O\left(1+ \frac{1}{n}(k + s \log s) \frac{\log W}{\log (2+k/s)} \right)$.
\end{proof}

\medskip

\noindent{\bf Lower Bound for Weighted SWOR:} We next present a lower bound on the message complexity of weighted SWOR, which follows from prior work on the message complexity of unweighted sampling without replacement. Let $s$ denote the desired sample size and $k$ the number of sites.

\begin{theorem}[\cite{TW11}]
\label{thm:unweighted-lb}
For a constant $q, 0 < q < 1$, any correct algorithm that continuously maintains an unweighted random sample without replacement from a distributed stream must send $\Omega \left (\frac{k \log (n/s)}{\log (1+(k/s))} \right )$ messages with probability at least $1-q$ and where the probability is taken over the algorithm's internal randomness, and where $n$ is the number of items.
\end{theorem} 

\begin{corollary}
\label{cor:weighted-lb}
For a constant $q, 0 < q < 1$, any correct algorithm that continuously maintains a weighted random sample without replacement from $\stream$ must send $\Omega \left (\frac{k \log (W/s)}{\log (1+(k/s))} \right )$ messages with probability at least $1-q$, where the probability is taken over the algorithm's internal randomness, and $W$ is the total weight of all items so far.
\end{corollary}

\begin{proof}
This lower bound follows since unweighted SWOR is a special case of weighted SWOR. We consider an input stream of $W$ items, each with a weight of $1$, and apply Theorem~\ref{thm:unweighted-lb}.
\end{proof}


 
\section{Weighted SWOR via Precision Sampling}\label{sec:SWORexpo}
We now present an algorithm for weighted SWOR on distributed streams. 
Our algorithm utilizes the general algorithmic framework of \textit{precision sampling}, introduced by Andoni, Krauthgamer, and Onak \cite{andoni2010streaming}. The framework of the algorithm is relatively straightforward. When an update $(e_i,w_i)$ is received at a site, the site generates a ``key'' $v_i = w_i/t_i$, where $t_i$ is a generated exponential random variable. The coordinator then keeps the stream items $(e_i,w_i)$ with the top $s$ largest keys $v_i$. We first state a result on exponential scaling that allows for the basic correctness of our sampling algorithm. 

\begin{proposition}[\cite{nagaraja2006order}, Equation 11.7 and Remark 1]\label{prop:nagaraja}
    Let $\mathcal{S} = \{(e_1,w_1),(e_2,w_2),\dots,(e_n,w_n)\}$ be a set of items, where each item has an identifier $e_i$ and a weight $w_i$. Suppose for each $i \in [n]$ we generate a key $v_i = w_i / t_i$, where the $t_i$s are i.i.d. exponential random variables with rate $1$ (i.e., with pdf $p(x) = e^{-x}$ for $x\geq 0 $). 
    For $k \in [n]$, let the \textit{anti-rank}s $D(k)$ be defined as the random variable indices such that $v_{D(1)} \geq v_{D(2)} \geq \dots \geq v_{D(n)}$. For $s \leq n$, let $S(s) \subset [n]$ be the set of items $(e_i,w_i)$ such that $i = D(k)$ for some $k \in [s]$ (i.e. $v_i$ is among the top $s$ largest keys). Then,
    
    \begin{itemize}
        \item $S(s)$ is a weighted SWOR from the set of items $\mathcal{S}$. 
        \item We have the distributional equality:
    \[   v_{D(k)} =	\left(\sum_{j=1}^k \frac{E_j}{\sum_{q=j}^{n} w_{D(q)}}	 	\right)^{-1} \]
    where the random variables $E_1,\dots,E_n$ are i.i.d. exponential random variables with rate $1$, and are independent of the anti-rank vector $(D(1),D(2),\dots,D(n))$.
    \end{itemize}
\end{proposition}

The above remark demonstrates that taking the items with the top $s$ largest keys indeed gives a weighted sample without replacement. The distributional equality given in the second part of Proposition \ref{prop:nagaraja} will be used soon in our analysis. Note that while the above results require continuous random variables to be used, we show that our results are not effected by limiting ourselves to a machine word of precision when generating the exponentials (Proposition \ref{prop:bitcomplex}). Thus our expected message complexity and runtime take into account the complexity required to generate the exponentials to the precision needed by our algorithm.

To reduce the number of messages sent from the sites to the coordinator, each site locally filters out items whose keys it is sure will not belong to the globally $s$ largest keys. In order to achieve this, our algorithm divides the stream into \textit{epochs}. 
The coordinator continuously maintains a threshold $u$, equal to the value of the smallest key $v$ of an item $(e,w,v)$ held in its sample set $S$, where $S$ always holds the items $(e_i,w_i,v_i)$ with the $s$ largest values of $v_i$. For $r = \max\{2,k/s\}$, whenever $u \in [r^j,r^{j+1})$, the coordinator declares to all sites that the algorithm is in epoch $i$ (at the beginning, the epoch is $0$ until $u$ first becomes equal to or larger than $r$). Note that this announcement requires $k$ messages to be sent from the coordinator, and must be done once per epoch.

If the algorithm is in epoch $j$, and a site receives an update $(e_i,w_i)$, it generates key $v_i$ and sends $(e_i,w_i,v_i)$ to the coordinator if and only if $v_i > r^j$. The coordinator will add the update $(e_i,w_i,v_i)$ to the set $S$, and if this causes $|S| > s$, it removes the item $(e,w,v)$ from $S$ which has the smallest key $v$ of all items in $S$. This way, at any time step $t$, the coordinator holds the items with the $s$-largest key values $v_i = w_i/t_i$, and outputs the set $S$ as its weighted sample. 


A complication in analyzing this algorithm is that it is not possible to directly connect the total weight received so far to the number of epochs that have progressed. For instance, it is not always true that the larger the total weight, the more epochs have elapsed (in expectation). For instance, if a few (fewer than $s$) items with a very large weight are received, then the total weight received can be large, but the $s$th largest key is still $0$, so that we are still in the zeroth epoch. 

To handle this situation, we introduce a \textit{level set} procedure. The idea is to \textit{withhold} heavy items seen in the stream from being sent to the sampler until they are no longer heavy. Here, by releasing an update $(e_i,w_i)$ to the sampler we mean generating a key $v_i$ for the item $(e_i,w_i)$ and deciding whether to accept $(e_i,w_i,v_i)$ into the sample set $S$. Specifically, we only release a heavy item to the sampler when its weight is no more than a $1/4s$ fraction of the total weight released to the sampler so far. We now defined the \textit{level} of an update $(e,w)$ below (Definition \ref{def:levelset}).

\begin{definition}\label{def:levelset}
The level of an item $(e,w)$ is the integer $j \geq 0$ such that $w \in \left[r^j,r^{j+1}\right)$. If $w \in \left[0,r\right)$, we set $j=0$. 
\end{definition}

For $j > 0$, the \textit{level set} $D_j$ will consist of the first $4rs$ items in the stream that are in level $j$.
We do this for all $j \geq 0$, but note that we can clearly stop at $j = \log(W)/\log(r)$, and we do not need to explicitly initialize a level set until at least one item is sent to the set. The coordinator stores the set $D_j$. As long as $|D_j| < 4rs$, if a site received an item $(e,w)$ that is in level $j$, the item is sent directly to the coordinator without any filtering at the site, and then placed into $D_j$. We call such a message an ``early'' message, as the item $(e,w)$ is withheld from the sampling procedure. We call all other messages which send an item $(e,w,v)$ from the site to the coordinator ``regular'' messages. Similarly, we call an update that resulted in an early message as an ``early update'', and all other updates as ``regular updates''. Note that a regular update may not result in a regular message, if the key of the regular update is smaller than the threshold for the current epoch.

We call the level set $D_j$ \textit{unsaturated} if $|D_j| < 4rs$ at a given time, and \textit{saturated} otherwise. Each site stores a binary value $\texttt{saturated}_j$, which is initialized to $\false$, indicating that $|D_j| < 4rs$. Once $|D_j| = 4rs$, the level set is saturated, and the coordinator generates keys $v_i = w_i/t_i$ for all items $(e_i,w_i) \in D_j$. For each new item-key pair, it then adds $(e_i,w_i,v_i)$ to $S$ if $v_i$ is in the top $s$ largest keys in $S \cup \{(e_i,w_i,v_i)\}$. At this point, the coordinator announces to all the sites that the level set $D_j$ has been saturated (who then set $\texttt{saturated}_j = \true$), and thereafter no early updates will be sent for this level set.

Note that in this procedure, the set $S$ will only be a weighted sample over the updates in level sets that are saturated. Since our algorithm must always maintain a weighted sample over \textit{all} stream items which have been received so far, we can simply simulate generating the keys for all items in the unsaturated level sets $D_j$, and take the items with the top $s$ largest keys in $S \cup (\cup_{j \geq 0} D_j)$ as the weighted sample (see the description of this in the proof of Theorem \ref{thm:maincorrect}). 
The algorithm for weighted SWOR is described in Algorithm~\ref{algo:swor-site} (algorithm at the site), and Algorithms~\ref{algo:swor-coord}, \ref{algo:add-to-sample} (algorithm at the coordinator). 

We now observe the following result of this level-set procedure. 


\begin{lemma}
At any time step, let $(e_1,w_1),\dots,(e_t,w_t)$ be all items so far that are in saturated level sets. For any $i \in [t]$, $w_i \leq \frac{1}{4s}\sum_{p \in [t]} w_p$. 
\end{lemma}

\begin{proof}
The reason is that for each item $(e_i,w_i)$ in a saturated level set, there are at least $4rs$ items in the same level set, whose weights are within a factor of $r$ of $w_i$. Thus $w_i$ can be no more than $\frac{1}{4s}$ of the total weight of this level set, and hence no more than $\frac{1}{4s}$ of the total weight. 
\end{proof}



\begin{algorithm}
\caption{Weighted \swor: Algorithm at Site $i$}
\label{algo:swor-site}

\tcp{Initialization}
\For{each level $j \geq 0$}{$\saturated_j \gets \false$}
\tcp{we will have that $u_i = r^j$ whenever the algorithm is in epoch $j$}
$\thresh_i \gets 0,$ $r \leftarrow \max\{2,k/s\}$\;
\tcp{End Initialization}

\While{$\true$}
{
  \If{receive item $(e,w)$}
  {
     \tcp{$e$ an item id and $w$ a positive weight}
     Let $j$ be the level set of $(e,w)$, i.e. the integer $j$ such that $w \in [r^j, r^{j+1})$. If $w \in (0,1)$, we set $j=0$\;

     \eIf{$\saturated_j = \false$}
     {
        Send $(\early, e,w)$ to the coordinator
     }
     { \tcp{$\saturated_j = \true$}
        Let $t$ be an exponential random variable \;
     $v \leftarrow w/t $\;
         \If{$v > \thresh_i$}
         {Send $(\regular, e,w,v)$ to the coordinator}
      }
  }

  \If{receive (``level saturated'', $j$) from coordinator}
  {
     Set $\saturated_j \gets \true$\;
  }

  \If{receive (``update epoch'', $r^j$) from coordinator}
  {
     $u \leftarrow r^j$\;
  }
}
\end{algorithm}




\begin{algorithm}
\caption{Weighted \swor: Algorithm at Coordinator}
\label{algo:swor-coord}

\tcp{Initialization}
\For{each level $j \geq 0 $}
{
  $\saturated_j \gets \false$\;
  $D_j \gets \emptyset$ \tcp{set of early messages in level $j$}
}
$\thresh \gets 0$  \tcp{$s$th max of keys of regular items}
$S \gets \emptyset$ \tcp{the current sample at the coordinator}
$r \leftarrow \max\{2,k/s\}$\;
\tcp{End Initialization}

\While{\true}
{
  \If{receive $(\early,e,w)$ from site $i$}
  {
       Let $j$ be the level set of $(e,w)$, i.e. the integer $j$ such that $w \in [r^j, r^{j+1})$. If $w \in (0,1)$, we set $j=0$\;
 Generate exponential $t$ with rate $1$\;
        Generate key $v = w/i$\;
          Add $(e,w,v)$ to $D_j$\;
     \If{$|D_j| \geq 4sr$}
     {
        \For{ each $(e',w',v') \in D_j$:}{
        
        \texttt{Add-to-Sample}$(\sample, (e',w',v'))$\;
        }
        $D_j \gets \emptyset$,     $\saturated_j \gets \true$\;
            Broadcast (``level saturated'', $j$) to all sites\;
     }
  }   
  
  \If{receive $(\regular,e,w,v)$ from site $i$}
  {
     \If{$v > \thresh$}
     {
         \texttt{Add-to-Sample}$(\sample,(e,w,v))$\;
     }
  }
  
  \If{receive query for sample of $s$ items}
  {Return items with the $s$ largest keys in $S \cup (\cup_{j}D_j)$}
}
\end{algorithm}

\begin{algorithm}
\caption{Weighted \swor: Add-to-Sample$(S, (e_i,w_i,v_i))$}
\label{algo:add-to-sample}
\KwIn{$S$ is the current sample at the coordinator, and $ (e_i,w_i,v_i)$ the item to be inserted}
$S \leftarrow S \cup \{(e_i,w_i,v_i)\}$\;
\If{($|S| > s$)}{
Let $v_{j_{\min}} = \arg \min_{v_t} \{v_t \; | \; (e_t,w_t,v_t) \in S\}$ \;
$S \leftarrow S \setminus \{(e_{j_{\min}},w_{j_{\min}},v_{j_{\min}} )\}$\;}
$u_{old} \leftarrow u$\;
$\thresh \leftarrow \min_{(e,w,v) \in S} v$\;
   \If{$\thresh \in [r^j,r^{j+1})$ \textbf{and} $u_{old} \notin [r^j,r^{j+1})$ for some $j$}
  {
     Broadcast (``update epoch'', $r^j$) to all sites
  }
\end{algorithm}

\subsection{Analysis}
We now analyze the message complexity of the algorithm. Let $r = \max\{2,k/s\}$, and let $u$ be the $s$-th largest value of the keys $v_j$'s given to the coordinator. As described, we define an epoch $i$ as the sequence of updates such that $u \in [r^i,r^{i+1})$. Set $W = \sum_{i=1}^n w_i$, where $n$ is the total number of updates at the end of the stream.  

Let $(e_1,w_1),\dots,(e_n,w_n)$ be the set of all stream items. For the sake of analysis, we consider a new ordering of the stream items, which corresponds to the order in which the keys $v_i$ are generated for the items $e_i$ in our algorithm. In other words, we assume the items are ordered  $(e_1,w_1),\dots,(e_{n'},w_{n'})$, such that the key $v_i$ for $e_i$ was generated before the key for $e_{i+1}$, and so on. This reordering is simply accomplished by adding each regular item to the ordering as they arrive, but withholding each early item, and only adding an early item to the ordering once the level set $D_j$ that it was sent to is saturated. Note that at the end of the stream, some level sets may not be saturated, thus $n' \leq n$. The remaining $n-n'$ items will have already been sent to the coordinator as early messages, and since no exponentials $t_i$ will have been generated for them when the stream ends, we can ignore them when analyzing the message complexity of the regular stream items. Note that this ordering is a deterministic function of the original ordering and weights of the stream, and has no dependency on the randomness of our algorithm. 

Now let $(e_{p_1},w_{p_1}), (e_{p_2},w_{p_2}), \dots, (e_{p_\tau},w_{p_\tau})$ be the subsequence of the regular items in  $(e_1,w_1),\dots,(e_{n'},w_{n'})$, so that $p_1 \leq p_2 \leq \dots \leq p_{\tau}$. 
We group the regular items $e_{p_1}, e_{p_2}, \dots, e_{p_\tau}$ into sets $\Omega_i^j$ as follows. For any $t \in [n']$, let $W_t$ be the total weight seen in the stream up to time $t$ under the new ordering. Thus $W_t = \sum_{i=1}^t w_i$.
We know by construction of the level sets, at any point in time that if $e_{p_i}$ is a regular item then $w_{p_i} \leq \eps_0 W_{p_i - 1}$, where $\eps_0 = \frac{1}{4s}$. For any $i \geq 1$, let $\Omega_j$ be the set of all regular items $(e_{p_i},w_{p_i})$ such that $W_{p_i } \in (s r^{j-1}, sr^{j})$. 
Note that by definition, $\Omega_j = \emptyset$ for all $j > z$, where $z =\lceil \log(W/s)/\log(r)\rceil$.
We now further break $\Omega_i$ into blocks $\Omega_{i}^1,\Omega_i^2,\dots,\Omega_i^{q_i}$, such that $\Omega_{i}^t$ consists of all regular items $(e_{p_i},w_{p_i})$ such that $W_{p_i} \in (s r^{i-1} (1+\eps)^{t-1},s r^{i-1} (1+\eps)^{t})$, for some value $\eps = \Theta(1)$ which we will later fix. Note that $q_i < \eps^{-1}\log(r)$.

Let $Y_i$ indicate the event that the $i$-th regular item caused a message to be sent, where we order the items via the new ordering as above. If $\p$ is the total number of regular items, then we would like to bound $\ex{\sum_{i=1}^{\p}Y_i}$. Note that the property of an item being regular is independent of the randomness in the algorithm, and is defined deterministically as a function of the ordering and weights of the stream. Thus $\p$ is a fixed value depending on the stream itself. Then $\ex{Y_i} = \sum_{j=1}^\infty \ex{Y_i | \mathcal{E}_{i,j}} \pr{  \mathcal{E}_{i,j}}$ where $\mathcal{E}_{i,j}$ is the event that we are in epoch $j$ when the $i$-th update is seen. Note that since $\mathcal{E}_{i,j}$ only depends on the values $t_1,\dots,t_{i-1}$, it follows that $\ex{Y_i | \mathcal{E}_{i,j}} = \pr{t_i < w_i/r^j} \leq  w_i/r^j$. Thus
\[ \ex{Y_i} \leq w_i \sum_{j=1}^\infty r^{-j} \pr{  \mathcal{E}_{i,j}} \]
Our goal will now be to bound the above quantity. We will only concern ourselves with $\mathcal{E}_{i,j}$ when $e_i$ is a regular message, since the early messages will be sent to the coordinator deterministically.

To bound $\pr{\mathcal{E}_{i,j}}$, we must bound the probability that we are in an epoch $j$ much longer than expected. To do this, we develop a tail bound on the probability that the $s$-th largest key $v_i = w_i/t_i$ is much smaller than expected. To do so, we will first need the following standard tail bound on sums of exponential random variables.

\begin{proposition}\label{prop:MGF}
	Let $E_1,\dots,E_k$ be i.i.d. mean $1$ exponential random variables, and let $E = \sum_{i=1}^k E_i$. Then for any $c \geq 1/2$, 
		\[\pr{E > c k} < \lambda e^{-Cc }		\]
		where $\lambda,C > 0$ are fixed, absolute constants.
\end{proposition}
\begin{proof}
The moment generating function of an exponential $E_i$ with mean $1$ is given by $\ex{e^{tE_i}} = \frac{1}{1-t}$ for $t < 1$. Thus $\ex{e^{tE}} = \ex{e^{t \sum_{i=1}^k E_i }} = \prod_{i=1}^k \ex{e^{tE_i}} = (\frac{1}{1-t})^k$. Setting $t=1/2$, then by Markov's inequality, $\pr{E > ck} < \pr{e^{tE} > e^{tck}} \leq  \frac{2^k}{ e^{1/2ck} } \leq e^{-1/2 c k + k} =\lambda e^{-\Omega(c)}$ for any $c \geq 1/2$ and a fixed constant $\lambda>0$ as needed. 
	\end{proof}

We now introduce our main tail bound on the behaviour of the $s$-th largest key $v_i$.

\begin{proposition}\label{prop:weights}
	Let $w_1,\dots,w_t$ be any fixed set of weights, and let $W = \sum_{i=1}^t w_i$. Fix $\ell \in [t]$, and suppose that $w_i \leq \frac{1}{2\ell}W$ for all $i \in [t]$. Let $v_i = w_i/t_i$ where the $t_i$'s are i.i.d. exponential random variables. Define the \textit{anti-ranks} $D(k)$ for $k \in [t]$ as the random variables such that $v_{D(1)} \geq v_{D(2)} \geq \dots \geq v_{D(t)}$. Then for any $c \geq 1/2$, we have
	\[ \pr{v_{D(\ell)}  \leq   \frac{W}{c \ell}} \leq O(e^{-Cc})	\]
	where $C>0$ is a fixed constant.
\end{proposition}
\begin{proof}
We note that $1/v_i$ is distributed as an exponential with rate $w_i$. The exponential order statistics of independent non-identical exponential random variables were studied by Nagaraja \cite{nagaraja2006order}, who demonstrates that the $\ell$-th largest value of $v_k$ in $\{v_1,v_2,\dots,v_t\}$ is distributed as  (see Proposition \ref{prop:nagaraja}, or equation $11.7$ in \cite{nagaraja2006order}):
\begin{equation*}
       v_{D(\ell)} =	\left(\sum_{j=1}^\ell \frac{E_j}{W - \sum_{q=1}^{j-1} w_{D(q)}}	 	\right)^{-1} 
         \ge   \frac{W}{2} 	\left(\sum_{j=1}^\ell E_j	\right)^{-1} 
  \end{equation*}
	where the $E_i$'s are i.i.d. exponential random variables with mean $1$ that are independent of the anti-ranks $D(1),\dots,D(t)$. Here, for the inequality, we used the fact that each regular item is at most a $1/(2\ell)$ heavy hitter at every intermediary step in the stream. It follows that if $v_{D(\ell)}  \leq    W/(c\ell)$, then $	\sum_{j=1}^\ell E_j \geq  \ell c/2$ 
	Which occurs with probability $O(e^{-Cc})$ for some fixed constant $C>0$ by Proposition \ref{prop:MGF}.
\end{proof}

\begin{proposition}\label{prop:seriesbound}
	Let $(e_i,w_i)$ be any regular stream item, and let $a_i,b_i$ be such that $e_i \in \Omega_{a_i}^{b_i}$. Then we have: $$\sum_{j=1}^\infty r^{-j} \pr{  \mathcal{E}_{i,j}}  \leq O\left(r^{-a_i+2} e^{-C\frac{(1+\eps)^{b_i-1}}{(1+\eps_0)}} + r^{-a_i + 1}\right)$$
	 where $C > 0$ is some absolute constant.
\end{proposition}
\begin{proof}
	First note for $j \geq a_i - 1$, we have $\sum_{j\geq k_i}^\infty r^{-j} \pr{  \mathcal{E}_{i,j}} \leq 2r^{-a_i +1}$. So we now bound $\pr{\mathcal{E}_{i,j}}$ for $j < a_i$. Let $v_q = w_q/t_q$ for $q \in [i-1]$, and let $D(p)$ be the anti-ranks of the set $\{v_1,\dots,v_{i-1}\}$. We note then that if $\mathcal{E}_{i,j}$ holds, then $v_{D(s)} < r^{j+1}$ by definition. Thus $\pr{\mathcal{E}_{i,j}} \leq \pr{v_{D(s)} \leq r^{j+1}}$. Note that by definition of the level sets $\Omega^{b_i}_{a_i}$, the total weight of all items seen up to time $i-1$ is $W_{i - 1} > W_{i}\frac{1}{(1+\eps_0)} > sr^{a_i - 1} \frac{(1+\eps)^{b_i-1}}{(1+\eps_0)}$.   Here we have used the fact that, by construction of the ordering, no item $w_i$ has weight greater than $\eps_0 W_i$ where  $\eps_0 = \frac{1}{4s}$. This fact will also be needed to apply Proposition \ref{prop:weights}.
	
	It then suffices to bound $\pr{v_{D(s)} < r^{j+1}} =\pr{v_{D(s)} <  \frac{1}{c} W_{i -1} /s}$, where  $c = W_{i-1} /(r^{j+1} s) \geq \frac{(1+\eps)^{b_i-1}}{(1+\eps_0)} r^{a_i - j - 2}$. Note that since $j \leq i - 2, b_i \geq 1$ and $\eps_0 = \frac{1}{4s} < 1/2$, we have $c > 1/2$. So by Proposition \ref{prop:weights}, $\pr{\mathcal{E}_{i,j}} \leq \pr{v_{D(s)} <  \frac{1}{c} W_{i-1} /s } =  O(e^{-C c}) = O(e^{-C\frac{(1+\eps)^{b_i-1}}{(1+\eps_0)} r^{a_i - j - 2}})$, where $C > 0 $ is some absolute constant. Thus 
	\[r^{-j} \pr{  \mathcal{E}_{i,j}}  \leq O\left(e^{-C\frac{(1+\eps)^{b_i-1}}{(1+\eps_0)} r^{a_i - j - 2}} r^{-j}  \right).\]
	If $j = a_i-2$, this is $O(r^{-a_i+2} e^{-\frac{C(1+\eps)^{b_i-1}}{(1+\eps_0)}})$.  In general, if $j = a_i - 2 - \ell$ for any $\ell\geq 1$, then this bound becomes: $$O\left(r^{-a_i+2 + \ell} e^{-\frac{C(1+\eps)^{b_i-1}}{(1+\eps_0)}}r^{\ell } \right) = O\left(r^{-a_i+2 - \ell} e^{-\frac{C(1+\eps)^{b_i-1}}{(1+\eps_0)}}\right)$$ where here we used the fact that $(e^{-x})^y = O(e^{-y}x^{-2})$ for any $x \geq 2$ and $y \geq \tau$ where $\tau > 0$ is some constant. Thus $\sum_{\ell=0}^{a_i - 1} r^{-a_i + 2 - \ell}\pr{\mathcal{E}_{i,j}} = O(r^{-a_i+2 } e^{-C\frac{(1+\eps)^{b_i-1}}{(1+\eps_0)}})$ since $r \geq 2$, which completes the proof.  
\end{proof}

\begin{lemma}\label{lem:regstreambound}
	Let $\p$ be the number of regular items in the stream.	If $Y_{1},Y_2,\dots$ are such that $Y_i$ indicates the event that the $i$-th regular stream item (in the ordering defined above) causes a message to be sent, then:
	\[\ex{\sum_{i=1}^{\p} Y_i}  =O\left(sr \frac{\log(W/s)}{\log(r)} \right).\]
\end{lemma}

\begin{proof}
	Let $\mathcal{W}_i^j$ be the weight in set $\Omega_i^j$. In other words, $\mathcal{W}_i^j = \sum_{w_{p_t} \in \Omega_i^j} w_{p_t}$. Recall we set $z =\lceil \log(W/s)/\log(r)\rceil$, and note that $\Omega_i = \emptyset$ for $i > z$. Also note that by construction we have $W_i^j  <  s r^{i-1}(1+\eps)^{j}$. Recall $q_i$ is the number of sets of the form $\Omega_i^j$ for some $j$, and that $q_i \leq O(\eps^{-1} \log(r))$ by construction. 
Using Proposition \ref{prop:seriesbound}, we have: 
\begin{equation*}
      \begin{split}
      \ex{\sum_i Y_i} &\leq \sum_{i=1}^z \sum_{j=1}^{q_i} \sum_{e_{p_i} \in \Omega_{i}^j} w_{p_i}  O\left(r^{-i+2} e^{-\frac{C(1+\eps)^{j-1}}{(1+\eps_0)}} +r^{-i + 1}\right) \\
      &\leq \sum_{i=1}^z \sum_{j=1}^{q_i} \mathcal{W}_{i}^j  O\left(r^{-i+2} e^{-\frac{C(1+\eps)^{j-1}}{(1+\eps_0)}}+ r^{-i + 1}\right) \\
      &\leq O(s) \sum_{i=1}^z\sum_{j=1}^{q_i}  (1+\eps)^{j}  \left(r e^{-\frac{C(1+\eps)^{j-1}}{(1+\eps_0)}} + 1\right)\\
       &= O(rs) \sum_{i=1}^z\sum_{j=1}^{q_i}  (1+\eps)^{j}   e^{-\frac{C(1+\eps)^{j-1}}{(1+\eps_0)}} \\ &  +  O(s)\left(\sum_{i=1}^z\sum_{j=1}^{q_i}  (1+\eps)^{j}   \right)\\
      \end{split}
  \end{equation*}

	\noindent
	Setting $\eps = .5$, we have  $(1+\eps_0) < (1+\eps)$, this is
	\[ \leq O(sr) \sum_{i=1}^z \sum_{j=1}^{q_i}\frac{1.5^{j}}{\exp\left(C(1.5)^{j-2}\right)}  +  s\left(\sum_{i=1}^z O(r)  \right)\]
	Now by the ratio test $\sum_{j=1}^{\infty} \frac{(1.5)^{j}}{\exp\left(C(1.5)^{j-2}\right)}$ converges absolutely to a constant (recalling $C>0$ is just some absolute constant), so the whole sum is:
	\begin{equation*}
 \leq sr\sum_{i=1}^z O(1) + s r\left(\sum_{i=1}^z O(1) \right) 
 =O\left(sr \frac{\log(W/s)}{\log(r)} \right)  
  \end{equation*}
	where we used $z =\lceil \log(W/s)/\log(r)\rceil$.
\end{proof}

We now bound the number of epochs used in the algorithm. We let the random variable $\zeta$ denote the total number of epochs in the algorithm. 

\begin{proposition}\label{prop:zeta}
	If $\zeta$ is the number of epochs in the algorithm, then 
	if $z =  \lceil \frac{\log(W/s)}{\log(r)} \rceil$, then 
	\[\ex{\zeta}  \leq 3\left( \frac{\log(W/s)}{\log(r)}+1\right). \]

\end{proposition}
\begin{proof}
After epoch $z+\ell$ for any $\ell \geq 0$, we have that $u > r^{\ell} W$, where $u$ is the value of the $s$-th largest key at the coordinator. Let $Y = |\{i \in [n] \; | \; v_i \geq r^\ell W		\}|$. Since the pdf of an exponential random variable is upper bounded by $1$, it follows that $\pr{v_i \geq r^{\ell} W/s	} = \pr{t_j  \leq s r^{-\ell} \frac{w_j}{W}	} \leq  s r^{-\ell} \frac{w_j}{W} $, thus $\ex{Y} \leq s r^{-\ell}	$. Thus $\pr{\zeta \geq z + \ell} \leq \pr{Y \geq s} \leq r^{-\ell}$, 
	where the last inequality follows by a Markov bound. Thus
		\begin{equation*}
      \begin{split}
 \ex{\zeta} &\leq  z + \sum_{\ell \geq 1}(z + \ell) \pr{\zeta \geq z + \ell} \\
 & \leq  z + \sum_{\ell \geq 1}(z + \ell) r^{-\ell} \leq 3z. \\
      \end{split}
  \end{equation*}
%
\end{proof}

\begin{lemma}\label{lem:message}
	The total expected number of messages sent by the algorithm is $O\left(sr \frac{\log(W/s)}{\log(r)} \right)$, where $r = \max\{2,k/s\}$. Thus this can be rewritten as 
	\[O\left( k \frac{\log(W/s)}{\log(1+k/s	)} \right). \]
\end{lemma}
\begin{proof}
	For each level set $B_t$ for $t < \log(W/s)/\log(r)$, at most $4 rs + k$ messages are sent, corresponding to the $4rs$ messages sent by sites to saturate the set, and then $k$ messages coming from the reply from the coordinator to all $k$ sites announcing that the set $B_t$ is full. This gives a total of $(4rs + k)\log(W/s)/\log(r)$ messages. Finally, there are at most $s$ items with weight greater than $W/s$, and thus we pay one message for each of these when they arrive at a site. The level sets corresponding to values greater than $W/s$ will never be saturated, so the total message complexity to handle the early messages is $O(sr\log(W/s)/\log(r))$ as needed. 
	
	Next, we pay $k$ messages at the end of every epoch. Thus if $\zeta$ is the number of epochs in the algorithm, the expected number of messages sent due to the end of epochs is $$\ex{\sum_{i=1}^\zeta k} = k\ex{\zeta} < k\frac{\log(W/s)}{\log(r)} =O( s r \frac{\log(W/s)}{\log(r)})$$ by Proposition \ref{prop:zeta}. 	Finally, the expected number of messages sent due to regular stream items is $O(sr \log(W/s)/\log(r) )$ due to Lemma \ref{lem:regstreambound}, which completes the proof of the message complexity.
\end{proof}

\begin{proposition}\label{prop:spacetime}
    The algorithm described in this section can be implemented, without changing its output behavior, to use $O(s)$ memory and $O\left( k \frac{\log(W/s)}{\log(1+k/s	)} \right)$ expected total runtime at the coordinator, and $O(1)$ memory and $O(1)$ processing time per update at each site. Note that all four of these bounds are optimal.
\end{proposition}
\begin{proof}
We first consider the sites.
Note that each site only generates a random variable, checks a bit to see if a level set is saturated, and then makes a single comparison to decide whether to send an item, all in $O(1)$ time. For space, note that each site only needs to store a bit that determines whether level set $D_j$ is full for $j \leq \log(W)/\log(r)$. For all level sets $j$ with  $j \geq \log(W)/\log(r)$, no item will ever arrive in this level since no item has weight more than $W$. Thus, to store the bits $\texttt{saturated}_j$, it suffices to store a bit string of length $\log(W)/\log(r)$, which is at most $O(1)$ machine words.

We now consider the coordinator. For space, we demonstrate how the level-set procedure can be carried out using less than $O(s r\log(W)/\log(r))$ space (which is the total numer of items that can be sent to level sets). For each item that arrives at a level set, we generate its key right away. Now note that the items in the level sets with keys that are not among the $s$ largest keys in all of the level sets will never be sampled, and thus never change the set $S$ or the behavior of the algorithm, thus there is no point of keeping them. So at any time, we only store the identities of the items in the level sets with the top $s$ keys (and which level they were in). Call this set $S_{level}$ To determine when a level set becomes saturated, we keep an $O(\log(rs))$-bit counter for each level set $D_j$, which stores $|D_j|$. Note that we only need to store this for the levels $j \leq \log(W/s)/\log(r)$, since at most $s$ items have weight more than $W/s$, and such level sets will never be saturated. 
When a level set $D_j$ becomes saturated, for each item $(e,w,v) \in S_{level}$ such that $(e,w)$ is in level $j$, we send this item to the main sampler (i.e., we decide whether or not to include it in $S$ based on the size of its key). The result is the same had the entire level set been stored at all times, and only $O(s)$ machine words are required to store $S_{level}$, and $O( \log(sr)\log(W)/\log(r)) = O(\log(s) \log(W))$-bits for the counters, which is $O(\log(s)) < O(s)$ machine words, so the total space is $O(s)$ machine words.

For the runtime of the coordinator, each early item and regular item is processed in $O(1)$ time. When a set is saturated, it takes at most $O(s)$ work to send these new items into the main sampler, and this occurs at most $O(\log(W/s)/\log(r))$ times, for a total of $O(s \log(W/s)/\log(r))$ work. Since every other message the coordinator recieves requires $O(1)$ work, the result follows from Lemma \ref{lem:message}.
%
%
%
\end{proof}

We now remark that, up until this point, we have not considered the bit complexity of our messages, or the runtime required to generate the exponentials to the required precision. We address this issue now. 
\begin{proposition}\label{prop:bitcomplex}
    The algorithm for weighted SWOR can be implemented so that each message requires an expected $O(1)$ machine words, and moreover, for any constant $c \geq 1$ uses $O(1)$ machine words with probability at least $1-W^{-c}$. Each exponential can be generated in time $O(1)$ in expectation, and, for any constant $c \geq 1$, time $O(1)$ with probability at least $1-W^{-c}$.
\end{proposition}
\begin{proof}
To see this, we note that a site only needs to generate enough bits to determine whether a given key is large enough to be sent. Recall that the quantile function for the exponential distribution is given by $F^{-1}(p) = -\ln(1-p)$ and, since $1-p$ and $p$ are distributed the same for a uniform variable, an exponential variable $t_i$ can be generated by first generating a uniform random variable $U$, and outputting $-\ln(U)$ \cite{knuth1997art}. Thus, if the algorithm is in epoch $j$, one can simply generate $U$ bit by bit, until one can determine whether $w_i/t_i > r^j$ or not. So each bit of $U$ generated cuts the remaining probability space by a factor of $2$. So for any threshold $\tau$, it requires only $O(1)$ bits to be generated in expectation to determine whether $-\ln(U) < \tau$, and since the probability space is cut by a constant factor with each bit, only $O(\log(W))$ bits (or $O(1)$ machine words) are needed with high probability in $W$. Thus the coordinator generates the message in expected $O(1)$ time and with high probability, and the size of the message is an expected $O(1)$ machine words with high probability. Similarly, when the coordinator decides whether to accept a sample into the set, it again needs only to check the exponential against a threshold, which requires generating at most an additional expected $O(1)$ bits and with high probability, which completes the proof. 
\end{proof}

\begin{theorem}\label{thm:maincorrect}
	The algorithm of this section for weighted sampling without replacement uses an expected $O\left( k \frac{\log(W/s)}{\log(1+k/s	)} \right)$ messages 
	and maintains continuously at every point in the stream a uniform weighted sample size $s$ of the items seen so far in the stream. The space required for the coordinator is $O(s)$ machine words, and the space required for each site is $O(1)$ words. Each site requires $O(1)$ processing time per update, and the total runtime of the coordinator is $O\left( k \frac{\log(W/s)}{\log(1+k/s	)} \right)$.
\end{theorem}
\begin{proof}The message complexity follows from Lemma \ref{lem:message}, and the space and time complexity follow from Proposition \ref{prop:spacetime}. The bit complexity issues which arise from dealing with exponentials are dealt with in Proposition \ref{prop:bitcomplex}. For correctness, note that at any point in time, the coordinator maintains two kinds of sets: the set $S$ which consists of the top $s$ samples, and the level sets $D = \cup_{j \geq 0} D_j$. To obtain a true weighted SWOR of the stream up to any point $i$, the coordinator can simply generate an exponential $t_j$ for each $(e_j,w_j) \in D$, and set $v_j = w_j / t_j$. It then constructs $D' = \{(e_j,w_j,v_j)| (e_j, w_j) \in D\}$, and returns the identifiers $e_j$ in $D' \cup S$ with the $s$ largest keys $v_j$. The result is that at any point in time, the coordinator can output the set $S$ of the entries $(e_j,w_j,v_j)$ with the top $s$ largest values of $v_j$.
	
	Let $\Delta_\ell$ be the identifiers $e_j$ with the top $\ell$ values of $v_j$. Since $1/v_j$ is exponentially distributed with rate $w_j$, for any $\ell \geq 1$, if $v_j \notin \Delta_{\ell-1}$ then it follows that the probability that $1/v_j$ is the $\ell$-th smallest (or $v_j$ is the $\ell$-th largest) is $\frac{w_j}{\sum_{e_t \notin \Delta_{\ell-1}} w_t}$ (via Proposition \ref{prop:nagaraja}). Thus the coordinator indeed holds a uniform weighted sample of the stream continuously at all points in time. Note that this sample may contain items in a level set $D_j$ which has not yet been filled at the time of query. However, this does not affect the behavior of the algorithm, since an exponential can be generated for such an item early and used as the sample (see proof of Proposition \ref{prop:spacetime}). By this, we mean that when the algorithm is queried for a weighted sample $S$ over \textit{all} stream items seen so far, it can simply generate a key for each item that is held in an unsaturated level set $D_j$, and keep the items with the top $s$ largest keys in $S \cup (\cup_{j \geq 0} D_j)$. Note, however, that the true set $S$ will not be modified by this procedure, and the actual items will not be sent into the set $S$ to potentially be sampled until the level set $D_j$ is full. 
		\end{proof}

\section{Tracking Heavy Hitters with Residual Error}
\label{sec:reshh}

In this section we demonstrate that our sampling algorithm results in a new algorithm for continuously monitoring heavy hitters with a \textit{residual} error guarantee.
This is also known as the heavy hitters tracking problem.  We show an   $O\left(\left(\frac{k}{\log(k)} + \frac{\log(1/(\eps \delta))}{\eps} \right) \log(\eps W) \right)$ upper bound for the message complexity of the problem, along with a nearly tight lower bound $\Omega\left(\left(\frac{k}{\log(k)} + \frac{1}{\eps} \right) \log(\eps W) \right)$.

 While the monitoring of heavy hitters has been studied substantially in the streaming and distributed streaming literature, to date no distributed algorithm has been designed which obtains heavy hitters with a residual error guarantee. We first formalize these variations  of the heavy hitters problem. For a vector $x \in \R^n$,  and any $t \in [n]$, let $x_{\tail(t)} \in \R^n$ be the vector which is equal to $x$ except that the top $t$ largest coordinates $|x_i|$ of $x$ are set to $0$. The following is the standard notion of the heavy hitters tracking problem.

\begin{definition}\label{def:hhnormal}
Let $\stream = (e_1,w_1),\dots,(e_n,e_n)$, and let $x^t \in \R^n$ be the vector with $x_i = w_i$ for $i \leq t$, and $x_i = 0$ otherwise. Then an algorithm $\mathcal{P}$ solves the $(\eps,\delta)$ heavy hitters tracking problem if, for any fixed $t \in [n]$, with probability $1-\delta$, it returns a set $S$ with $|S| = O(1/\eps)$  such that for every $i \in [t]$ with $x_i \geq \eps \|x^t\|_1$, we have $x_i \in S$.
\end{definition}
Now, we introduce the variant of the heavy hitters tracking problem that requires residual error. 

\begin{definition}\label{def:hhresidual}
Let $\stream = (e_1,w_1),\dots,(e_n,e_n)$, and let $x^t \in \R^n$ be the vector with $x_i = w_i$ for $i \leq t$, and $x_i = 0$ otherwise. Then an algorithm $\mathcal{P}$ solves the $(\eps,\delta)$ heavy hitters tracking problem \textit{with residual error} if, for any fixed $t \in [n]$, with probability $1-\delta$ it returns a set $S$ with $|S| = O(1/\eps)$ such that for every $i \in [t]$ with $x_i \geq \eps \|x_{\tail(1/\eps)}^t \|_1$, we have $x_i \in S$.
\end{definition}

Note that the residual error guarantee is strictly stronger, and captures substantially more information about the data set when there are very large heavy items. 

\begin{theorem}
\label{thm:heavy-hitter}
There is a distributed streaming algorithm $\mathcal{P}$ which solves the $(\eps,\delta)$ heavy hitters problem with residual error, and sends an expected $O\left(\left(\frac{k}{\log(k)} + \frac{\log(1/(\eps \delta))}{\eps} \right) \log(\eps W) \right)$ messages. The algorithm uses $O(1)$ space per site, $O(1)$ update time per site, $O(\frac{\log(1/(\delta\eps))}{\eps})$ space at the coordinator, and $O\left(\left(\frac{k}{\log(k)} + \frac{\log(1/(\eps \delta))}{\eps} \right) \log(\eps W) \right)$ overall runtime at the coordinator.
\end{theorem}
\begin{proof}
To obtain the bound, we run our weighted SWOR of Theorem \ref{thm:maincorrect} with $s = 6\frac{\log(1/(\delta\eps))}{\eps}$, where $C>0$ is a sufficently large constant. Fix a time step $t$, and let $S^t$ be the sample obtained by the weighted SWOR at that time. We know that $S^t$ is a weighted SWOR from $(e_1,w_1),\dots,(e_t,w_t)$. Let $x^t$ be the vector corresponding to these weights, and let $T \subset [t]$ be the set of coordinates such that $x_i \geq \eps \|x_{\tail(1/\eps)}^t \|_1$. Let $T_0 \subset T$ be the subset of $i$ such the $x_i^t$ is one of the $1/\eps$ largest values in $x^t$. 
Now fix any $i \in T$, and break $S$ into blocks $S^1,S^2,\dots,S^{2\log(1/(\delta\eps))}$, each of size $3/\eps$. Now for each $j \in [2\log(1/(\delta\eps))]$, there are at least $2/\eps$ items which are not in $T_0$ (which follows since $|T_0| \leq 1/\eps$). Conditioned on a sample $s \in S^j$ not being in $T_0$, since $x_j \geq \eps \|x_{\tail(1/\eps)}^t \|_1$ by assumption, the probability that $s =i$ is at least $\eps$. Thus the probability that $i \notin S^j \leq (1-\eps)^{2/\eps}  < 1/2$. Repeating $2\log(1/\eps)$ times, the probability that $i \in S$ is at least $1 - (\frac{1}{2})^{2\log(1/(\delta\eps))} < 1- (\eps\delta)^2$. We can then union bound over all $|T| < 2/\eps$ items, so that $T \subset S$ with probability at least $1-\delta$. To restrict the set $S$ to $O(1/\eps)$ items, we simply order the items in $S$ by weight and output the top $2/\eps$, which will contain $T$ if $S$ does, which completes the proof of correctness.  The remainder of the Theorem then follows from the complexity bounds of Theorem \ref{thm:maincorrect}.
\end{proof}

\subsection{Lower Bound for Tracking Heavy Hitters}

We now demonstrate an $\Omega(k \log(W)/\log(k) + \log(W)/\eps))$ lower bound on the expected message complexity of any distributed algorithm that monitors the heavy hitters in a weighted stream (Definition \ref{def:hhnormal}). Since the residual error guarantee is strictly stronger, this lower bound extends to monitoring heavy hitters with residual error. 
\begin{theorem} \label{thm:hhlowerbound}
Fix any constant $0 < q < 1$, and let $\eps \in (0,1/2)$. Then any algorithm which $(\eps/2,\delta) = (\eps,q^2/64)$ solves the heavy hitters tracking problem (Definition \ref{def:hhnormal}), must send at least $\Omega(k \log(W)/\log(k) \allowbreak + \eps^{-1} \log(W))$ messages with probability at least $1-q$ (assuming $1/\eps < W^{1-\xi}$ for some constant $\xi > 0$). In particular, the expected message complexity of such an algorithm with $\delta = \Theta(1)$ is $\Omega(k \log(W)/\log(k) + \eps^{-1} \log(W))$.
\end{theorem}
\begin{proof}
Create $s = \Theta(\frac{1}{\eps}\log(W)) $ global stream updates $(e_i,w_i)$, such that $w_i = (1+\eps)^i \eps$ and $w_0 = 1$. Then note for each $i \geq 0$, we have that $w_i$ is an $\eps/(1+\eps) > \eps/2$ heavy hitter in $w_1,w_2,\dots,w_i$. Note that the total weight of the stream is $\sum_{i=1}^sw_i = W$. 

Now consider any algorithm $\mathcal{P}$ for $(\eps/2,q^2/64)$ heavy hitter tracking, and at any time step $t$ let $S$ denote the set of heavy hitters held by $\mathcal{P}$,). On a given time step $t \in [n]$, let $S_t$ denote the value of $S$ at time $t$. and let $R$ denote the concatenation of all random coin flips used by the algorithm $\mathcal{P}$ for both the sites and the coordinator. 
We first claim the following:
\begin{claim} \label{claimhh} Let $0 < q < 1$ be any constant. 
Suppose $\mathcal{P}$ is a randomized algorithm which $(\eps/2,q/64)$-solves the heavy hitter tacking problem. Suppose the set of heavy hitters it maintains at all times is called $S$. Then there is a constant $C' = C'(q) > 0$ such that:
\begin{itemize}
    \item The set $S$ changes at least $\frac{C'}{\eps}\log(n)$ times with probability at least $1-q$.
\end{itemize}
\end{claim}
\begin{proof}

Consider the above stream instance.
Note that on each time step, the new update $(e_i,w_i)$ becomes an $\eps/2$ heavy hitter, and therefore should be accepted into the set $S$. Note moreover that at any time step $t$, the total weight is $(1+\eps)^t$. 
Suppose there is a randomized algorithm $\mathcal{P}$ which for each $t \in T$, succeeds in outputting all the $\eps$ heavy hitters in a set $S$ at time $t$ with probability $1-q/64$. Let $X_t$ be a random variable that indicates that $\mathcal{P}$ is correct on $t$. We have $\ex{X_t} > 1-q/64$, and clearly $\ex{X_t X_{t'}} \leq 1$ for any $t,t'$. Thus Var$(\sum_{i=1}^{\frac{1}{\eps}\log(n)} X_i) \leq \frac{1}{\eps}\log(n)(1-q/64) + \frac{1}{\eps^2}\log^2(n) - \frac{1}{\eps^2}\log^2(n)(1-q/64)^2 \leq (q/30)\frac{1}{\eps^2}\log^2(n)$ for $n$ larger than some constant. By Chebyshev's inequality:
\[\pr{ \left|\frac{1}{\eps}\log(n)(1-q/64) - \sum_{i \leq \eps^{-1}\log(n)} X_i \right| >\frac{1}{q} \sqrt{q/30}\frac{\log(n)}{\eps} } <  q \] 
Thus with probability at least $1-q$, $\mathcal{P}$ is correct on at least a  $\frac{C'}{\eps}\log(n)$ number of the points in $T$, for $C' = (1-q/64 - \frac{1}{\sqrt{30}} -) > 1/2$. Now suppose that, conditioned on this, the algorithm changed $S$ fewer than $\frac{C'}{\eps}\log(n)$ times. Note that every time a new item $(e_i,w_i)$ arrives, if $\mathcal{P}$ is correct on time $i$, it must be the case that $(e_i,w_i) \in S$ once update $i$ is done processing. Thus if $S$ did not change, but $(e_i,w_i) \in S_t$, this means $(e_i,w_i) \in S_{t-1}$ -- the coordinator contained $e_i$ in its set before the item arrived! By having $e_i$'s being $O(\log(W))$ bit identifiers and randomizing the identities in the stream, it follows that the probability that this could occur at any time step $t$ is less than $(W/\eps) (1/\poly(W)) < W^{-100}$ for $e_i$'s with $O(\log(W))$ bits and a sufficiently large constant. So we can safely condition on this event not occurring without affecting any of the claims of the theorem. It follows that $S$ must change on every time step $t$ on which it is correct, which completes the first claim that $S$ must change at least $\frac{C'}{\eps}\log(n)$ times with probability at least $1-q$. 
\end{proof}

Observe that the lower bound of  of $\Omega(\log(W)/\eps)$ messages being sent with probability at least $1-q$ follows directly from the last claim.
For the $\Omega(k \log(W)/\log(k))$ lower bound, we construct a new weighted stream. 
Define $\eta = \Theta(\frac{\log(W)}{\log(k)})$ \textit{epochs}, as follows. In the $i$-th epoch, each site $j$ receives \textit{one} weighted update $(e_i^j, k^{i})$. Note then at the very beginning of the $i$-th epoch (for any $i$), the total weight seen in the stream so far is at most $2k^{i}$. Thus the first update $(e_i^j,k^i)$ that arrives in the $i$-th epoch will be a $1/2$ heavy hitter, and thus must be accepted into $S$ for the protocol to be correct on that time step.

Now consider any epoch $i$, and let $X_i$ be the number of sites which sent or received a message from the coordinator in epoch $i$. Note that $X_i$ is a lower bound on the number of messages sent in epoch $i$.  Let $X_i^j$ indicate that site $j$ sent or recieved a message from the coordinator in epoch $i$, so that $X_i = \sum_{j \in [k]} X_i^j$. We claim $\ex{X_i} = \Omega(k)$. To demonstrate this, consider the first site $j^*$ which recieves an update in epoch $i$. Since item $(e_i^{j^*}, k^i)$ is immediately a $1/2$ heavy hitter, it must be sent to be correct. Since the protocol is correct on the $ki+1$'st time step with probability at least $1-q/64$, it follows that site $j^*$ must send a message after receiving its update in epoch $i$ with probability $1-q/64$, so $\ex{X_i^{j^*}} > 1-q/64$.

But note that when a site $j$ receives its update in the $i$-th epoch, if it has not communicated to the coordinator since the start of the epoch, it does not know the order in which it recieved the item. In particular, it does not know whether it was the first site to receive an item.  Since the randomized protocol must be correct on any advesarial fixing of the ordering, it follows that if site $j$ does not know if it recieved the first update, it must nevertheless send it with probability at least $1-q/64$ to be a correct protocol. In other words $\ex{X_i^{j}} > 1-q/64$ \textit{for all} $j \in [k]$, from which $\ex{X_i} > (1-q/64)k$, and therefore $\ex{\sum_{i=1}^\eta X_i} > (1-q/64)k$ follows. Thus $\ex{k \eta - \sum_{i=1}^\eta X_i} < k \eta q/64$, and so by a Markov bound over the last expectation, $\pr{\sum_{i=1}^\eta X_i < (1-1/64) k \eta } < q$. Since $X_i$ lower bounds the number of messages sent in the $i$-th epoch, this completes the proof.
\end{proof}

\section{$L_1$ tracking}
\label{sec:l1}

In this section, we demonstrate how our sampling algorithm from Section \ref{sec:SWORexpo} can be used to design a message-optimal $L_1$ tracking algorithm, which will close the complexity of this problem. We first recall the definition of an $L_1$ tracking algorithm.

\begin{definition}
\label{def:L1}
Given a distributed weighted data stream $\mathcal{S} = (e_1,w_1),(e_2,w_2),\dots, \allowbreak (e_n,w_n)$ and parameters $\eps,\delta > 0$, a distributed streaming algorithm $(\eps,\delta)$ solves the $L_1$ tracking problem if the coordinator continuously maintains a value $\widetilde{W}$ such that, at any fixed time $t \in [n]$, we have $\widetilde{W} = (1 \pm \eps) W_t$ with probability $1-\delta$, where  $W_t = \sum_{i=1}^t w_i$.
\end{definition}

In \cite{cmyz12}, an $L_1$ tracker was given that had $O(k \log(W) \log(1/\delta))$ expected messages. An improved bound of $O((k + \sqrt{k} \log(1/\delta)/\eps)\log(W))$ expected messages was then provided in  \cite{huang2012randomized}, along with a lower bound of $\Omega(\frac{\sqrt{k}}{\eps}\log(W))$ for all $k \leq \frac{1}{\eps^2}$. We remark that while the bounds of \cite{huang2012randomized} are stated as $O(\sqrt{k}/\eps \log(W))$ (for constant $\delta$), the actual expected message complexity of their algorithm is $O((k + \sqrt{k}\log(1/\delta)/\eps)\log(W) )$ (in \cite{huang2012randomized} the authors assume that $k < 1/\eps^2$). Nevertheless, the aforementioned bound is, up to this point, the current best known distributed algorithm for $L_1$ tracking. 

In this section, we give tight bounds for the case $k > 1/\eps^2$. More specifically, we prove an $L_1$ tracking algorithm with expected message complexity $O(  \frac{k\log(\eps W)}{\log(k)} + \allowbreak \frac{\log(\eps W) \log(1/\delta)}{\eps^2} )$. In the following section, we also provide an $\Omega(k \frac{\log(W)}{\log(k)})$ lower bound for the problem. We remark that the results of \cite{huang2012randomized} show an  $\Omega(\frac{\sqrt{k}}{\eps} \log(W))$ lower bound assuming that $k < \frac{1}{\eps^2}$, and they give a matching upper bound in this regime. When  $k > \frac{1}{\eps^2}$, the lower bound becomes $\Omega(\frac{1}{\eps^2}\log(W))$, which can be seen by simply applying the same lower bound on a subset of $1/\eps^2$ of the sites. Then when $k > 1/\eps^2$, if $k \log(W)/\log(k) <\log(W)/\eps^2$, our upper bound is $O(\log(W)/\eps^2)$, which is tight. Otherwise, our upper bound is $O(k \log(W)/\log(k))$, which matches our lower bound of $\Omega(k\log(W)/\log(k))$. Thus our protocol is optimal whenever $k > 1/\eps^2$. Since tight upper and lower bounds were known for the case of $k < 1/\eps^2$ by \cite{huang2012randomized}, this closes the complexity of the problem.


We also note that our lower bound assumes $1/\eps < W^{1-\xi}$ for some fixed constant $\xi >0$. Thus our lower bound does not contradict our upper bound of $\Omega(k\log(\eps W)/\log(k) + \log( \eps W)/\eps^2)$ (for $\delta = \Theta(1)$). Observe that for $1/\eps = \Theta(W)$, one can simply send every stream update for $O(W)$ communication, which is also clearly a lower bound since any stream update not immediately sent would cause the coordinator to no longer have a $(1 \pm \eps)$ approximation. 

Note for an $L_1$ tracking such that the coordinator has a value $\widetilde{W}$ such that $(1-\eps) W_t \leq \widetilde{W} \leq (1 + \eps) W_T$
for \textit{every} step $t=1,2,\dots,n$ with probability $1-\delta$, our algorithm uses $O ( \frac{k\log(W/\eps)}{\log(k)} +\allowbreak \frac{\log(W/\eps) \allowbreak \log(\frac{\log(W)}{\eps \delta})}{\eps^2} )$ expected message complexity. This second result simply comes from union bounding over the $\frac{\log(W)}{\eps \delta}$ points in the stream where the $L_1$ increases by a factor of $(1+\eps)$.

The known upper and lower bounds for $L_1$ tracking on a single time step with failure probability $\delta = \Theta(1)$ are summarized below. Note again, constant $L_1$ tracking at all points in the stream is accomplished in all cases by setting $\delta = \log(W)/\eps$, and note that the $\log(1/\delta)$ does not multiply the $O(k\log(W)/\log(k))$ term in our upper bound. With the combination of the upper and lower bounds of \cite{huang2012randomized}, this completely closes the complexity of the problem.

\begin{center}
       \begin{tabular}{|c|c|}
       \hline
Citation & Message Complexity (upper or lower \\
 &  bound) \\ \hline
 \cite{cmyz12} + folklore   & $O\left(\frac{k}{\eps} \log(W) \right)$\\\hline
  \cite{huang2012randomized}   & $O\left(k \log(W) + \frac{\sqrt{k}\log(W)\log(1/\delta)}{\eps}\right)$ \\ \hline
  This work & $O \left(   \frac{k\log(\eps W)}{\log(k)} + \frac{\log(\eps W) }{\eps^2} \right)$ \\\hline  
  \cite{huang2012randomized}    & $\Omega\left(\frac{\sqrt{ \min\{ k, 1/\eps^2\} }}{\eps}\log(W) \right)$  \\\hline
  This work & $\Omega\left(k \frac{\log( W)}{\log(k)}\right)$ \\ \hline 
\end{tabular}
\end{center}

Our algorithm for $L_1$ tracking is described formally in Algorithm \ref{alg:L1} below. To show correctness, we first prove Proposition \ref{prop:Expobound2}, which is a standard tail bound for sums of exponential random variables. We utilize the techniques from our weighted SWOR algorithm to assign an key to each item $(e,w)$ in the stream.  Then, by the concentration of sums of exponential random variables, we will demonstrate that the $s$-th largest key can be used to obtain a good approximation to the current $L_1$ of the stream, where  $s = \Theta(\frac{1}{\eps^2}\log(\frac{1}{\delta}))$

\begin{Frame}[\textbf{Algorithm  \ref{alg:L1}: Tracking $L_1$}]
\label{alg:L1}
\texttt{Input:} Distributed stream $\mathcal{S}$\\
\texttt{Output:} a value $\tilde{W}$ such that $\tilde{W} = (1 \pm \eps)W_t$ at any step time $t$ with probability $1-\delta$.\\
\textbf{Initalization:}\\
Instantiate weighted SWOR algorithm $\mathcal{P}$ of Theorem \ref{thm:maincorrect} with $s = \Theta(\frac{1}{\eps^2}\log(\frac{1}{\delta}))$ \\
\textbf{On Stream item $(e,w)$}
\begin{enumerate}
\item  Duplicate $(e,w)$, a total of $\ell =\frac{s}{2 \eps}$ times, to yield $(e^1,w),(e^2,w),\dots,(e^\ell,w)$. Insert each into the algorithm $\mathcal{P}$. 
\end{enumerate}
\textbf{On Query for $L_1$ at time $t$:} 
\begin{enumerate}
    \item Let $u$ be the value stored at the coordinator of the $s$-th largest key held in the sample set $S$. 
  
    \item Output $\tilde{W} = \frac{s u}{\ell} $.
\end{enumerate}
\end{Frame}

\begin{proposition}\label{prop:Expobound2}
    Let $E_1,E_2,\dots,E_s$ be i.i.d. exponential random variables with mean $1$. Then for any $\eps > 0$, we have:
    \[\pr{|\sum_{j=1}^s E_j  - s| > \eps s    } < 2e^{-\eps^2 s/5} \]
\end{proposition}
\begin{proof}
The moment generating function of $\sum_{j=1}^s E_j$ is given by $(\frac{1}{1-t})^s$ for $t < 1$, 

	\begin{equation*}
      \begin{split}
\pr{\sum_{j=1}^s E_j > (1+\eps) s} &< \frac{(\frac{1}{1-t})^s}{e^{t(1+\eps) s}}  \\
 & \leq \exp\left( -t(1+\eps)s +s \log(\frac{1}{1-t})     \right)\\
      \end{split}
  \end{equation*}
Setting $t = \eps$ and taking $\eps < 1/2$ we obtain
	\begin{equation*}
      \begin{split}
&\leq \exp\left( -\eps^2s - \eps s -s \log(1-\eps)     \right)\\
 &\leq \exp\left( -\eps^2s - \eps s +s( \eps + \eps^2/2 + \eps^3/3 + \dots)       \right)\\
 &\leq \exp\left( -\eps^2s/5   \right)\\
      \end{split}
  \end{equation*}
as needed. Next, we have
	\begin{equation*}
      \begin{split}
 \pr{\sum_{j=1}^s E_j < (1-\eps) s} &<  \pr{e^{-t\sum_{j=1}^s E_j} > e^{-t (1-\eps) s}}  \\ & \leq \frac{(\frac{1}{1+t})^s}{\exp(-t(1-\eps) s)} \\
 & \leq \exp\left(   t(1-\eps) s - s \log(1 + t)  \right) \\
      \end{split}
  \end{equation*}
setting $t= \eps$:
	\begin{equation*}
      \begin{split}
 & \leq \exp\left(  -  \eps^2 s + \eps s  - s(\eps + \eps^2/2 + \eps^3/3 + \dots)  \right) \\
 & \leq \exp\left(  -  \eps^2 s/5   \right) \\
      \end{split}
  \end{equation*}
and union bounding over the upper and lower tail bound gives the desired result.
\end{proof}

\begin{theorem}\label{thm:l1main}
    There is a distributed algorithm (Algorithm \ref{alg:L1}), where at every time step the coordinator maintains a value $\widetilde{W}$ such that at any fixed time $\tau \in [n]$, we have 
    $$(1-\eps) W_\tau \leq \widetilde{W} \leq (1 + \eps) W_\tau$$
    with probability $1-\delta$, where $W_\tau = \sum_{i=1}^\tau w_i$ is the total weight seen so far at step $\tau$. The expected message complexity of the algorithm is $O\left( \left(\frac{k}{\log(k)} + \eps^{-2} \log(1/\delta)  \right)\log\left( \frac{ \eps W}{\log(\delta^{-1})}\right)\right)$
\end{theorem}
\begin{proof}
If $W_\tau$ is the total weight seen so far in the stream, then after duplication the weight of the new stream $\mathcal{S}'$ that is fed into $\mathcal{P}$ is $W_\tau \frac{2s}{\eps}$. Call the total weight seen up to time $\tau$ of the new stream $\overline{W}_\tau$. The duplication has the effect that at any time step $\tau$ in the original stream, every item $(e^i_j,w_j) \in \mathcal{S}'$, corresponding to a duplicate of $(e^i_j,w_j) \in \mathcal{S}$, is such that $w_j \leq \frac{\eps}{2s} \overline{W}_\tau$. Thus once an update to $\mathcal{S}$ is finished being processed, no item in $\mathcal{S}'$ is more than an  $\frac{\eps}{2s}$ heavy hitter. Since there are $\ell > s$ duplications made to each update, the level sets used in $\mathcal{P}$ will be immediately saturated the moment the first update in $\mathcal{S}$ with a weight in that level set arrives. Thus there is never any intermediate point in the stream $\mathcal{S}$ where some of the weight of the stream is being withheld in the level sets used by $\mathcal{P}$. Thus the total weight of the stream seen so far at any given time has already been sent to the sample, so  each item $(e_i,w_i)$ seen so far in the stream has had a key $v_i$ generated for it. 

Now the coordinator holds the value $u$ such that $u$ is the $s$-th largest value in $\{v_1,v_2,\dots,v_{\ell \tau}\}$, where $\ell$ was the number of duplications. Here $v_i = w_i / t_i$, where $t_i$ is an exponetial variable and $w_i$ is a weight of one of the duplicates sent to the stream $\mathcal{S}'$. By results on the order statistics of collections of exponential variables (\cite{nagaraja2006order}), we have the distributional equality
\[u =	\left(\sum_{j=1}^s \frac{E_j}{\overline{W}_\tau - \sum_{q=1}^{j-1} w_{D(q)}}	 	\right)^{-1}\]
where $E_1,E_2,\dots,E_s$ are i.i.d. exponential random variables, and $D(q)$ are the random variable indices such that $v_{D(1)} \geq V_{D(2)} \geq \dots \geq V_{D(\tau \ell)}$ (the $D(q)$'s are known as anti-ranks). Since $w_{D(q)} \leq\frac{\eps}{2s} \overline{W}_\tau$, it follows that 

\begin{equation*}
      \begin{split}
     u &=	\left( (1 \pm \eps) \sum_{j=1}^s \frac{E_j}{\overline{W}_\tau}	 	\right)^{-1}\\
     & =	(1 \pm O(\eps)) \overline{W}_\tau \left( \sum_{j=1}^s E_j	 	\right)^{-1}\\
        & =	(1 \pm O(\eps)) \ell W_\tau \left( \sum_{j=1}^s E_j	 	\right)^{-1}\\
      \end{split}
  \end{equation*}

Now by Proposition \ref{prop:Expobound2}, we have $\pr{|\sum_{j=1}^s E_j - s | > \eps s} \allowbreak < 2e^{-\eps^2 s/5} \allowbreak < \delta$, where here we set $s = 10 \log(\delta^{-1})/\eps^2$. Conditioned on this not occurring, we have $u = 	(1 \pm O(\eps)) \ell W_\tau \frac{1}{s}$, thus $u\frac{s}{\ell} = (1 \pm O(\eps)) W_{\tau}$ as required, and the expected message complexity follows by Theorem \ref{thm:maincorrect} setting $s = \Theta(\frac{1}{\eps^2}\log(\frac{1}{\delta}))$.
\end{proof}

\begin{corollary}
  There is a distributed algorithm where at every time step the coordinator maintains a value $\widetilde{W}$ such that 
    $$(1-\eps) W_\tau \leq \widetilde{W} \leq (1 + \eps) W_\tau$$
    \textit{for every} $\tau \in [n]$  
    with probability $1-\delta$. The expected message complexity of the algorithm is $$O\left( \left(\frac{k}{\log(k)} + \eps^{-2}\log(\frac{\log(W)}{\delta \eps}) \right)\log\left( \frac{ \eps W}{\log(\delta^{-1})}\right)  \right).$$
\end{corollary}
\begin{proof}
The result follows from running Theorem \ref{thm:l1main} with $\delta' = \log(W)/(\delta\eps)$, and union bounding over the $\log(W)/\eps$ time steps $t$ in the stream where $W_t$ increases by a factor of $(1+\eps)$.
\end{proof}

\subsection{Lower Bound for $L_1$ Tracking}
We now demonstrate an $\Omega(k \log(W)/\log(k) + \log(W)/\eps))$ lower bound on the expected message complexity of any distributed $L_1$ tracker. 
We remark again that our lower bound does not contradict out upper bound of $O(k\log(\eps W)/\log(k) + \log( \eps W)/\eps^2)$, since it requires $1/\eps < W^{1-\xi}$ for some constant $\xi$. 

Note that the lower bound holds for both weighted and unweighted streams, as the hard example in the proof is a stream where all weights are equal to $1$. We remark that the $\Omega(\log(W)/\eps)$ portion of the bound is only interesting when $k < \frac{1}{\eps}$, and in this regime a better lower bound of $\Omega(\frac{\sqrt{k}}{\eps} \log(W))$ is given by \cite{huang2012randomized}. We include this portion of the bound in the Theorem for completeness, but remark that the main contribution of the Theorem is the lower bound of $\Omega(k \log(W)/\log(k))$.  

\begin{theorem} \label{thm:l1lowerbound}
Fix any constant $0 < q < 1$. Then any algorithm which $(\eps,\delta) = (\eps,q^2/64)$ solves the $L_1$ tracking problem (Definition \ref{def:L1}), must send at least $\Omega(k \log(W)/\log(k)\allowbreak  + \eps^{-1} \log(W))$ messages with probability at least $1-q$ (assuming $1/\eps < W^{1-\xi}$ for some constant $\xi > 0$). In particular, the expected message complexity of such an algorithm with $\delta = \Theta(1)$ is $\Omega(k \log(W)/\log(k) + \eps^{-1} \log(W))$.
\end{theorem}
\begin{proof}
We proceed much in the way of Theorem \ref{thm:hhlowerbound}. 
We define $\eta = \Theta(\frac{\log(W)}{\log(k)})$ epochs as follows: at the end of the the $i-th$ epoch, for $i = 0,1,2,\dots,\eta -1$, there will have been exactly $k^i$ global stream updates processed since the start of the stream, each with weight $1$. Thus each epoch $i$ is deterministically defined, and contains $k^{i+1} - k^{i}$ global updates. Let the unweighted updates be $e_1,e_2,\dots,e_n$ (we can make these weighted by simply adding a weight of $1$. Here $n = W$, so the stream is both weighted and unweighted. In each epoch, we partition the updates arbitrarily over the sites $k$.

Now consider any algorithm $\mathcal{P}$ for $L_1$ tracking, and at any time step $t$ let $u$ denote the value held by $\mathcal{P}$ which tracks the $L_1$ (total weight seen so far in the stream). On a given time step $t \in [n]$, let $u_t$ denote the value of $u$ at time $t$, and let $R$ denote the concatination of all random coin flips used by the algorithm $\mathcal{P}$ for both the sites and the coordinator. 
We first claim the following:
\begin{claim} \label{claiml1} 
Let $0 < q < 1$ be any constant. 
Suppose $\mathcal{P}$ is a randomized algorithm which $(\eps,q/64)$-solves the $L_1$ tracking problem. Suppose the value it maintains for the $L_1$ approximation at all times is called $u$. Then there is a constant $C' = C'(q) > 0$ such that the following holds:
\begin{itemize}
    \item The value $u$ changes at least $\frac{C'}{\eps}\log(W)$ times with probability at least $1-q$.
\end{itemize}

\end{claim}
\begin{proof}
Consider the time steps $T= \{1,\texttt{rnd}(1+\eps),\texttt{rnd}((1+\eps)^2),\dots,\dots,n\}$, where the $L_1$ changes be a factor of $(1+\eps)$, where $\texttt{rnd}$ rounds the vaue to the nearest integer (note we can assume $n$ is a power of $(1+\eps)$ by construction). Suppose there is a randomized algorithm $\mathcal{P}$ which for each $t \in T$, succeeds in outputting a $(1\pm\eps)$ approximation to $t$ at time $t$ with probability $1-q^2/64$. Assuming that $1/\eps < W^{1-\xi}$ for some constant $\xi > 0$, we have $|T| > \frac{1}{\eps} \log(W)$ distinct such points. Let $X_t$ be a random variable that indicates that $\mathcal{P}$ is correct on $t$. We have $\ex{X_t} > 1-q/64$, and clearly $\ex{X_t X_{t'}} \leq 1$ for any $t,t'$. Thus Var$(\sum_{i=1}^{\frac{1}{\eps}\log(W)} X_i) \leq \frac{1}{\eps}\log(W)(1-q/64) + \frac{1}{\eps^2}\log^2(W) - \frac{1}{\eps^2}\log^2(W)(1-q/64)^2 \leq (q/30)\frac{1}{\eps^2}\log^2(n)$ for $n$ larger than some constant. By Chebyshev's inequality:
\[\pr{ \left|\frac{1}{\eps}\log(W)(1-q/64) - \sum_{i \leq \eps^{-1}\log(W)} X_i \right| >\frac{1}{\sqrt{q}} \sqrt{q/30}\frac{\log(W)}{\eps} } <  q \] 
Thus with probability at least $1-q$, $\mathcal{P}$ is correct on at least a  $\frac{C'}{\eps}\log(W)$ number of the points in $T$, for $C' = (1-q/64 - \frac{1}{\sqrt{30}} -) > 1/2$. Now suppose that, conditioned on this, the algorithm changed $u$ less than $\frac{C'}{\eps}\log(W)$ times. But each time $u$ changes, it can be correct for at most one of the values in $T$, which contradicts the fact that $\mathcal{P}$ is correct on at least a  $\frac{C'}{\eps}\log(W)$ of these points. So with probability at least $1-q$, the value of $u$ must change at least $\frac{C'}{\eps}\log(W)$ times, as needed. 

\end{proof}
Note that the lower bound of $\Omega(\log(W)/\eps)$ messages being sent with probability at least $1-q$ follows directly from the last claim. Note moreover, that our lower bound reduces to $\Omega(\log(W)/\eps)$ when $k < 8/\eps$, so we can now assume that $k >8/\eps$.

Now consider any epoch $i$, and let $X_i$ be the number of sites which sent or received a message from the coordinator in epoch $i$. Let $X_i^j$ indicate that site $j$ sent or recieved a message from the coordinator in epoch $i$, so that $X_i = \sum_{j \in [k]} X_i^j$. We claim $\ex{X_i} = \Omega(k)$.  To demonstrate this, first note that at the beginning of an epoch, the current $L_1$ of the stream is $k^i$. First condition on the event $\mathcal{E}_i$ that the estimate of the algorithm is correct at the beginning of the epoch. Note $\pr{\mathcal{E}_i} > 1-q/64$. By the law of total expectation: $\ex{X_i} \geq \ex{X_i \; | \; \mathcal{E}_i} (1 - q/64)$. 

Now consider the stream of updates $\sigma_i^j$ in epoch $i$ which send exactly $2k^i$ consecutive updates to each site $j$. Consider any stream $\sigma_i$ constructed by composing $\sigma_i^j$'s for any arbitrary permutation of the $j$'s.  Note that if site $j^*$ recieved the first set of updates in $\sigma_i$ and does not send a message after these updates, the algorithm will have an incorrect $L_1$ estimate at time $3k^i$. Since this does not occur with probability at least $1-q/64$, it follows that $\ex{X_i^{j^*}} \geq (1-q/64)$. But note that this must hold for all sites $j$, and once a site $j$ receives its set of $2k^i$ updates, since it could have been the first one to recieve the updates (unless it hears otherwise from the coordinator). In other words, since the randomized protocol must maintain the same guarantee for any ordering of the stream, it follows that every site $j$ must also send a message to the coordinator with probability at least $1-q/64$ after it sees its $2k^i$ updates (unless the coordinator sends a message to it first). Thus $\ex{X_i^{j}} \geq (1-q/64)$ for all $j$, which completes the claim that $\ex{X_i} \geq (1-q/64)k$, giving $\ex{\sum_{i=1}^\eta X_i} > k\eta(1-q/64)$. Thus $\ex{k \eta - \sum_{i=1}^\eta X_i} < k \eta q/64$, and so by a Markov bound over the last expectation, $\pr{\sum_{i=1}^\eta X_i < (1-1/64) k \eta } < q$. Since $X_i$ lower bounds the number of messages sent in the $i$-th epoch, this completes the proof.

\end{proof}

\section{Conclusions}
\label{sec:conc}
We presented message-efficient algorithms for maintaining a weighted random sample from a distributed stream. Our algorithm for weighted SWOR is optimal in its message complexity, space, as well as processing time. We also presented an optimal algorithm for $L_1$ tracking and the first distributed algorithms for tracking heavy hitters with residual error. One open question is whether we can obtain matching upper and lower bounds for weighted (or unweighted) sampling with replacement. While we can reduce weighted sampling with replacement to unweighted, known upper bounds are still a $\log(s)$ factor larger than the lower bounds.
Another problem is to extend our algorithm for weighted sampling to the sliding window model of streaming, where only the most recent data is taken into account for the sample. Finally, for the residual heavy hitters problem, can one remove the $\log(1/\epsilon)$ factor from our bounds?
\\\\
{\bf Acknowledgments}:
Rajesh Jayaram and David P. Woodruff thank the Simons Institute for the Theory of Computing where part of this work was done. They also acknowledge support by the National Science Foundation under Grant No. CCF-1815840. Srikanta Tirthapura acknowledges support by the National Science Foundation under grants 1527541 and 1725702.




\bibliographystyle{ACM-Reference-Format}

\bibliography{sampling}


\end{document}